 \documentclass[a4paper,11pt]{article}
\usepackage{amsmath}
\usepackage{amssymb}
\usepackage{graphicx}
\usepackage{appendix}
\usepackage{color}
\usepackage{times}
\usepackage{fullpage}
\usepackage{mathrsfs}

\newcommand{\remove}[1]{}

\newtheorem{proposition}{Proposition}%[section]

\newtheorem{corollary}{Corollary}%[section]

\newtheorem{definition}{Definition}%[section]

\newtheorem{lemma}{Lemma}%[section]

\newtheorem{theorem}{Theorem}%[section]

\newtheorem{remark}{Remark}%[section]

\newenvironment{proof}{\noindent{\bf Proof\@:}}{\hfill $\Box$\\}

\newenvironment{propositionproof}[1]{\noindent{\bf Proof of Proposition #1\@:}}{\hfill $\Box$\\}

\newcommand{\myshow}[1]{}

\date{\today}
\date{}
\title{Broadcasting colourings on trees.\\
A combinatorial view.\thanks{Supported by EPSRC grant EP/G039070/2 and DIMAP.}\\
}

\author{Charilaos Efthymiou\\
Goethe University, Mathematics Institute, Frankfurt 60054, Germany\\	
{\tt efthymiou@gmail.com}}

\begin{document}

\maketitle

 \begin{abstract}
The broadcasting models on a $d$-ary tree $T$ arise in many contexts such as discrete mathematics,
biology, information theory,  statistical physics and computer science. We consider the $k$-colouring 
model, i.e. the root of $T$ is assigned an arbitrary colour and, conditional on this assignment, 
we take a random colouring of $T$.  A basic question here is whether the information of the assignment 
at the root affects the distribution of the colourings at the leaves.  This is the so-called 
{\em reconstruction/non-reconstruction problem}. It is well known that $d/\ln d$ is a threshold function 
for this problem, i.e.
\begin{itemize}
 \item if $k\geq (1+\epsilon)d/\ln d$, then the colouring of the root has a vanishing effect
 on the distribution of the colourings at the leaves, as the height of the tree grows
 \item if $k\leq (1-\epsilon)d/\ln d$, then   the colouring of the root biases the distribution of the colouring
       of the leaves regardless of the height of the tree.
\end{itemize}
However, there is no apparent {\em combinatorial} reason why such a result should be true.

When $k\geq (1+\epsilon)d/\ln d$, the threshold implies the following: We can {\em couple} two broadcasting
processes that assign the root different colours such that the probability of having disagreement
at the leaves reduces with their distance from the root. It is natural to perceive such coupling as a mapping
from the colouring of the first broadcasting process to the colouring of the second one. In
that terms, here, we study how can we have such a mapping ``combinatorially''. 
Devising a mapping where the disagreements vanish as we move away from the root turns out
to be a non-trivial task to accomplish for any $k\leq d$.

In this work we obtain a coupling which has the aforementioned property for any 
$k>3d/\ln d$, i.e. much smaller than $d$.  Interestingly enough, the decisions that
we make in the coupling are somehow {\em local}. It is not clear clear whether such 
a coupling should be local for any $k$ down to $d/\ln d$. Finally, we relate our result to sampling 
$k$-colourings of sparse random graphs, with expected degree $d$ and $k\leq d$.
\end{abstract}

\section{Introduction}		

The broadcasting models on trees and the closely related reconstruction problem were 
originally studied in {\em statistical physics}. Since then they have found applications
in other areas including {\em biology} (in phylogenetic reconstruction \cite{OptPhylogeny,PhaseTransPhylog}), 
{\em communication theory} (in the  study of noisy computation \cite{IsingTree}). Very 
impressively, these models arise in {\em computer science} in the study of {\em  random
constraints satisfaction problems} such as random $k$-SAT, random graph colouring etc. 
That is, the models on trees seem to capture some of the most fundamental properties 
of the corresponding models on random (hyper)graphs, \cite{1RSBPaper}.

The most basic problem in the study of broadcasting models is to determine
the reconstruction/non-reconstruction threshold. I.e. whether the configuration 
of the root affects the distribution of the configuration of the leaves of
the tree. The transition from non-reconstruction to reconstruction can be
achieved by adjusting appropriately the parameters of the model. Typically,
this transition exhibits a {\em threshold behaviour}. So far, the main focus of
the study was to determine the precise location of this threshold for various models.

In this work, we focus on the colouring model on a $d$-ary tree. The reconstruction/non-reconstruction 
threshold for this model is known precisely \cite{InfFlowTrees,Semerjian,SlyRecon,TreeNonNaya}.
We investigate the phenomenon further by searching for a {\em combinatorial} reason
why the information decays in the non-reconstruction regime. Such an explanation, somehow, has been 
elusive when $k\leq d$. For the reconstruction regime combinatorial explanation is already known 
\cite{InfFlowTrees,Semerjian}.

Let us be more specific on what do we mean by combinatorial explanation. The threshold implies that
when $k\geq (1+\epsilon)d/\ln d$ we can couple two broadcasting  processes that assign the root different 
colours such that the probability of having disagreements at the leaves reduces as their distance from 
the root increases. 
It is  natural to perceive such coupling as a mapping from the colouring of the first broadcasting 
process to the colouring of the second one. In that  terms, here, we study how can we have such a 
mapping combinatorially.

We provide a coupling between two broadcasting processes which implies non-reconstruction
for $k$ well below $d$, i.e. for  $k>3d/\ln d$. It is based on describing a (combinatorial)
mapping between the colourings of two different broadcasting processes. It works inductively 
and considers two levels of the underlying tree each time.  E.g. given the colour assignments 
of the root in the two processes the coupling considers only colour choices for the vertices 
up to two levels below. The basic idea is to reveal partially some information for the decisions
of the two processes and investigates for which (small) subtrees of the root the colour assignments
at their leaves are identically distributed (conditional the revealed information).

Even though the coupling we present here is not  optimal, a lot of its basic ideas are quite 
natural. It seems reasonable to expect that an optimal coupling should adopt  a lot of them.
Finally, recent advances in sampling colouring algorithms (see \cite{mysampling}) relate this
coupling to sampling $k$-colourings of random graphs of expected degree $d$ when $k<d$ (see
Section \ref{sec:FurtherMot}).

\subsection{The model and the reconstruction problem}

The broadcasting models on a tree $T$ are  models in which information is sent from the
root over the edges to the leaves. We assume that the  edges represent noisy channels.
For some finite set of spins  ${\Sigma}=\{1,2,\ldots, k\}$, a configuration on $T$ is 
an element of $\Sigma^T$, i.e. it is an assignment of spins to the vertices of $T$. 
The spin of the root $r$ is chosen according to some initial distribution over $\Sigma$.  The information 
propagates along the edges of the tree as follows: There is a $k \times k$ stochastic  matrix 
$M$ such that if the vertex $v$ is assigned spin $i$, then its child $u$ is assigned spin $j$  
with probability $M_{i,j}$.

Our focus is on the {\em $k$-colouring} model (or $k$-state {\em Potts} model at zero temperature).
We assume that the underlying tree $T$ is a complete $d$-ary tree of {\em height} $h$
and for the matrix $M$  we have that
\begin{displaymath}
M_{i,j}=\left \{
\begin{array}{lcl}
\frac{1}{k-1} &\quad & \textrm{for $i\neq j$}\\
0 && \textrm{otherwise.}
\end{array}
\right.
\end{displaymath}
Broadcasting models give rise to Gibbs measures on trees. E.g. for the colouring model, assuming 
that the broadcasting process over $T$ starts with root $r$ coloured $i$, then the $k$-colouring 
we get after the processes has finished is a random $k$-colouring of $T$ conditional that $r$ is 
coloured $i$.

We let $L_h$ denote the {\em leaves}  of $T$. Also, we let $\mu_i$ denote the uniform
distribution over the $k$-colourings of $T$ conditional  that $r$ is assigned colour $i$.
Reconstructibility is defined as follows: 

\begin{definition}\label{def:Reconstuction}
For any $i,j\in [k]$ let $||\mu_i-\mu_j||_{L_h}$ denote the total variation distance
of the projections of $\mu_i$ and $\mu_j$ on $L_{h}$.  We say that a model is {\em reconstructible} on 
a tree $T$ if there exists $i,j\in [k]$ for which
\begin{displaymath}
\lim_{h\to \infty}||\mu_i-\mu_j||_{L_h}>0.
\end{displaymath}
When the above limit is zero for every $i,j$, then we say that the model has
{\em non-reconstruction}.
\end{definition}

\noindent
(Non)Reconstructibility expresses how information decays along the tree. As a matter of fact, 
non-reconstruction is equivalent to the mutual information between the colouring of root $r$ 
and that of $L_h$ is going to zero as $h$ grows (see \cite{Beating2ndEigen}).

When $T$ is infinite ($h\to \infty$)  non-reconstruction is equivalent to the Gibbs measure 
being {\em extremal}. That is, the distribution of the colouring at the root $r$ cannot be 
expressed as a convex combination of boundary conditions at the leaves of $T$ (see \cite{Georgii}).
For finite $h$, non-reconstruction implies that {\em typical} colourings of the leaves have 
a vanishing bias on the distribution of the colouring of $r$.

An early result about reconstruction/non-reconstruction problems on trees is the so called 
``Kesten-Stigum bound'' in \cite{Kesten}. The authors there show that reconstruction holds
when $\lambda^2 d>1$, where $\lambda$ is the second largest eigenvalue of $M$ in absolute value.
This bound is sharp for a lot of models, e.g. Ising model (see \cite{IsingTree}).  In \cite{Beating2ndEigen}
it was shown that there are models where the  Kesten-Stigum bound is not sharp, e.g. the 
binary models where $M$ is sufficiently asymmetric or the ferromagnetic $q$-state Potts model
with $q$ large. As far as the $k$-colouring model is regarded the  reconstruction 
threshold is known quite precisely. From \cite{InfFlowTrees,Semerjian,SlyRecon,TreeNonNaya} we derive 
the following theorem:

\begin{theorem}\label{thrm:NonRecon}
For fixed $\epsilon>0$ and sufficiently large $d$, the following is true 
for the $k$-colouring model on a $d$-ary tree $T$:
\begin{itemize}
\item If $k\geq (1+\epsilon)d/\ln d$, then the model is non-reconstructible.
\item If $k\leq (1-\epsilon)d/\ln d$, then the model is reconstructible.
\end{itemize}
\end{theorem}

\begin{remark} \em
The reconstruction bound is from \cite{InfFlowTrees,Semerjian}
and is based on analyzing  a simple reconstruction algorithm. As a matter
of fact the reconstruction condition there is more precise than that in
Theorem \ref{thrm:NonRecon}, i.e. it should hold $d>k[\ln k+\ln\ln k+1+o(1)]$.
\end{remark}

\begin{remark}\em
The non-reconstruction  bound is from \cite{SlyRecon, TreeNonNaya}.
The result in \cite{SlyRecon} provides a very precise condition for 
non-reconstruction, i.e. $d\leq k[\ln k+\ln\ln k+1-\ln 2-o(1)]$.
In \cite{TreeNonNaya} the reader can find further interesting results 
about the problem.
\end{remark}

\noindent
%From Theorem \ref{thrm:NonRecon} it is not completely obvious why non-reconstruction should hold
%for $k\geq (1+\epsilon)d$. To this end, 
Using the Coupling Lemma (see \cite{coupling-lemma})  with  Theorem \ref{thrm:NonRecon} we get the 
following corollary.

\begin{corollary}\label{cor:TheCoupling}
Consider a $d$-ary tree $T$ of height $h$. Assume that two broadcasting processes on $T$
assign the root different colours. For $\epsilon$ and $d$  as in Theorem \ref{thrm:NonRecon}
and $k=(1+\epsilon)d/\ln d$  there is a coupling for the two processes such that 
the following holds: The probability that there are leaves  with different colour 
assignments in the two processes reduces as $h$ increases.
\end{corollary}

\noindent
Somehow there is a rule which specifies how someone should correspond the choices of colourings 
in the first broadcasting process to the choices of the other one such that the probability of
having the leaves taking different colours reduces with their distance from the root.
Unfortunately, neither of \cite{SlyRecon, TreeNonNaya} casts a light on this question.
It turns out that devising such a coupling is far from trivial for any $k \leq d$.

Here we address the problem of constructing a coupling as specified in Corollary \ref{cor:TheCoupling}, 
based on  {\em local combinatorial} rules. By local we mean that once the first process decides on
the colouring of a {\em fairly small} set of vertices, then we should be able to know how the other process 
should colour the same set of vertices.  In particular, we provide the following result:
\\ \vspace{-.3cm}

\noindent
{\bf Main Result:} We construct a coupling of the processes in Corollary \ref{cor:TheCoupling}.
The coupling is combinatorial, local and implies non-reconstruction %has the properties specified in Corollary \ref{cor:TheCoupling} 
for any $k\geq (3+\epsilon)d/\ln d$,  where $\epsilon>0$ is fixed and $d$ is sufficiently large.

\paragraph{Notation.} We use small letters of the greek alphabet for the colourings of $T$, e.g.
$\sigma, \tau$. The capital letters denote random variables which take values over the colourings 
e.g. $X,Y$.  We let $\sigma_v$ denote the colour assignment of the vertex $v$ under the colouring 
$\sigma$. Similarly, the random variable $X(v)$ is equal to the colour assignment that $X$ specifies 
for the vertex $v$. For an integer $k>0$ we let $[k]=\{1,\ldots,k\}$.

\subsection{Further Motivation - Non Reconstruction in Random Graphs \& Sampling} \label{sec:FurtherMot}

It is believed that the non-reconstruction/reconstruction transition determines the {\em dynamic phase 
transition} for the $k$-colourings of the random graph $G(n,m)$. Where $G(n,m)$ denotes the random graph
on $n$ vertices and $m$ edges with $d$ denoting the expected degree, i.e. $d=2m/n$.
%In this context we take $d$ to be fixed.

The {\em dynamic phase transition} is related to the geometry of $k$-colourings of $G(n,m)$ and it 
was predicted by statistical physicists in \cite{1RSBPaper}, based on ingenious but mathematically 
non-rigorous arguments. Let us be more specific. For typical instances of $G(n,m)$, the chromatic number
$\chi$ is well known to be  $\chi\sim \frac{d}{2\ln d}$ (see \cite{GnmChromatic}). The 1-step Replica
Symmetry breaking hypothesis \cite{1RSBPaper} considers  the space of $k$-colourings
of $G(n,m)$ as $k$ varies from large to small and predicted the following phenomenon: For 
$k=(1+\epsilon)d/\ln d$ (i.e. greater than $2\chi$) all but a vanishing fraction of $k$-colourings
form a giant connected ball. That is, starting from any colouring we
can traverse the whole set of colourings in the ball by moving in steps. Each steps involves changing
only a very small -constant- number of colour assignments. However, for $k=(1-\epsilon)d/\ln d$ (e.g.
smaller than $2\chi$) the set of $k$-colouring shatters into exponentially many connected balls with 
each ball containing an exponentially small fraction of all $k$-colourings. Any two colourings in different 
balls are separated with linear hamming distance (for rigorous result about shattering see in \cite{ACO}).

\remove
{%\color{red}
The dynamic phase transition in $G(n,m)$ (roughly) coincides with the re\-con\-struction/non-re\-con\-struction 
transition of the colourings in a $d$-ary tree. Further investigation into this coincidence enabled the authors 
in \cite{GershMont} to developed a sufficient condition for the tree and random graph  reconstruction problem 
to coincide. In \cite{Restrepo} this condition was verified for symmetric models like colouring. 
}

It is believed that we can approximately  randomly colour $G(n,m)$ efficiently for $k$ down to the dynamic phase transition
threshold, i.e. $k=(1+\epsilon)d/\ln d$. %in  the whole non-re\-con\-struc\-tion regime. 
Recently, the author of this paper in \cite{mysampling} suggested a new algorithm for sampling colourings of 
$G(n,m)$ with constant expected degree. Inte\-re\-sting\-ly enough the accuracy of the algorithm depends directly 
on non-reconstruction conditions. The idea there is that we first remove edges of $G(n,m)$  until it becomes 
so simple that we can take a random colouring in polynomial time. Then, we {\em rebuild} the graph by adding 
the deleted edges one by one while at the same time we {\em update the colouring}. I.e. whenever a new edge is
inserted some vertices' colouring is updated so that the colouring of the resulting graph remains random. This  algorithm 
requires at least $(2+\epsilon)d$ colours. However, since its  accuracy depends on non-reconstruction conditions 
it is reasonable to expect that we can have an improvement by using even less colours. The algorithm does not 
exploit fully its dependency on non-reconstruction due to its colouring {\em update rule}. A new, improved, update
rule is needed. Such an improvement could possibly reduce the minimum number of colours that the algorithm requires 
down to $(1+\epsilon) d/\ln d$. Very good candidates for improved updating rules are couplings as the one we present
here. %Using this coupling for sampling is a highly non-trivial task.

\subsection{A basic description of the coupling.}

Consider two broadcasting processes, the first one $k$-colours $T$ as $X$ and the second
as $Y$. Assume that the root $r$ of $T$ is coloured  such that $X(r)=c$ and $Y(r)=q$ while 
$c\neq q$, for some $c,q\in [k]$.

Consider, first, the following recursive {\em naive coupling} of the two processes. Start from 
the root $r$ down to the leaves. For each coloured vertex $u\in T$ we colour its descendant $w$ 
by using {\em maximal} coupling. I.e. minimize the probability of $w$ to be disagreeing. 
If $X(u)\neq Y(u)$,  then  we have $X(w)\neq Y(w)$ only if $X(w)=Y(u)$ and $Y(w)=X(u)$.
On the other hand, if $X(u)=Y(u)$ then we always have $X(w)=Y(w)$.
It is not hard to see that $Pr[X(w)\neq Y(w)]=1/k$.

Clearly, when $k\leq d$, we expect that the naive coupling generates an ever increasing number of 
disagreeing vertices as it moves from the root down to the leaves.  As a matter of fact the
number of disagreeing vertices at each level grows as a {\em supercritical} branching process, 
i.e. the probability of having a disagreement at the leaves is strictly positive, regardless of their 
distance from the root.

Before introducing our coupling, consider the following notions. Let $N_i$ denote the 2 level subtree of
$T$ rooted at the $i$-th child of the root $r$. 
In the same setting as that in the naive coupling,  the colouring $X(N_i)$ is ``bad''  if $X(i)=q$ and
$i$ has a child $j$ such that $X(j)=c$. Similarly, $Y(N_i)$ is bad if $Y(i)=c$ and $i$ has a child $j'$ 
such that $Y(j')=q$.

In the naive coupling, $X(N_i)$ is bad if and only if $Y(N_i)$  is bad. For such a pair the identity 
coupling is precluded and the creation of disagreements is inevitable. That is, the naive coupling handles
the appearance of bad lists by coupling them together. Clearly this is not desirable. Especially, for $k\leq d$ 
the number of bad colourings $X(N_i)$, $Y(N_i)$  are  {\em ``too many''}.  This causes the ever
increasing number of disagreements of the naive coupling.
%In particular, we expect to have a lot of  $X(N_i)$s (or $Y(N_i)$s) which are bad.
%

%One idea to overcome this undesired phenomenon it the following one: Since identity coupling is precluded 
%between two bad colouring, try to find other colourings which can be coupled identically with a bad pair.
%avoid coupling bad colourings together. 
%
The coupling we propose here %follows exactly this idea. In particular it 
uses the following, not so obvious, observation to handle the bad lists:
Consider $X(N_j)$ conditional that (A) it is a bad and (B) there
is at least one colour that is not used by $X(N_j)$. Then, it is highly likely that there is a child 
of $r$, e.g. the vertex $s$, where $Y(N_s)$ satisfies the following two conditions: (A') The colour 
$Y(s)$ is not assigned to any child of $j$ under the colouring $X$ and (B') the colour 
$c$ is assigned to at least one child of $s$ under the colouring $Y$.
For such $X(N_j)$ and $Y(N_s)$, we can show the following:  The colour assignment of the children of $j$ in the first
process is identically distributed to that of the children of $s$ in 
the second process. 

Based on the above observation, the target now is to couple the colourings $X(N_i)$s and $Y(N_i)$s 
such that if $X(N_i)$ satisfies the conditions (A) and (B), then $Y(N_i)$ 
satisfies (A') and (B') and vice versa\footnote{I.e. the analogous conditions should hold for
bad $Y(N_j)$.}. Then, clearly,  we can couple  the colouring of children of the vertex $i$ 
identically. Let us remark that it is not completely trivial to ``aline'' these two 
different kinds of colouring in the coupling.

Working as described in the previous paragraph, the number of disagreements drops dramatically, 
compared to the naive coupling. As a matter of fact the number of disagreeing vertices grows as
a {\em subcritical} branching process, i.e. the probability of having disagreement at the leaves 
drops exponentially with their distance from the root.

\begin{remark}\em
The update rule in the sampling algorithm in \cite{mysampling}, somehow,  is based on 
what we call here {naive coupling}. 
\end{remark}

\section{Coupling}\label{sec:StronCoupling}

In this section we present the coupling in full detail. 
We let $\mu(\cdot)$ denote the uniform distribution over the $k$-colourings of $T$.
We consider two  broadcasting processes such that the first one assigns 
colour $c$ to the root while the second one assigns colour $q$. To avoid trivialities assume that $c\neq q$.
Finally, we let $X$, $Y$  be the colourings that the two processes assign to $T$, respectively. 
We proceed by introducing some useful concepts.

\subsection{Preliminaries}

Let $c_1, c_2,c_3\in [k]$ and, for $j=1,2,3$, let $L_j$  be a $d$-dimensional list
which contains colours in $[k]\backslash\{c_j\}$. For these three lists we have the following:

\begin{description}

\item [bad:] The pair $(L_{1}, L_{2})$  is called {\em bad}  if and only if $c_1\neq c_2$  and $c_2\in L_{1}$
while $c_1\in L_{2}$. 

\item [rescuable:] A bad pair $(L_{1},L_{2})$ is called {\em rescuable} if there is at least one 
colour in $[k]\backslash\{c_1,c_2\}$ that does not appear in both $L_{1}$ and $L_{2}$. 

\item [special:] Given that the pair $(L_{1},L_{2})$ is rescuable, the list $L_{3}$
is called ``special w.r.t. $L_{1}$'' if the  following holds:
\begin{enumerate}
\item $c_3$ does not belong to any of $L_{1},L_{2},L_{3}$ while $c_3\neq c_1$.
\item One of the following two holds:
\begin{enumerate}
\item $L_{3}$  contains $c_1$ but it does not  contain $c_2$
\item $L_{3}$  contains $c_2$ but it does not  contain $c_1$.
\end{enumerate}
\end{enumerate}

\item [good:] Given that the pair $(L_{1},L_{2})$ is rescuable, a list $L_{3}$  is called
``good w.r.t. $L_{1}$'' if it is  special  w.r.t. $L_{1}$ and the condition 2(b) holds.
\end{description}

\begin{definition}\label{def:RandomList}
For $c\in [k]$, we let ${\lambda}_c$ denote the uniform distribution over the $d$
dimensional lists of colours which do not contain the color $c$.
\end{definition}

\begin{lemma}\label{lemma:RescuabelVsGood}
Let $c,q,s\in [k]$ such that they are different with each other. Let $L_1$, $L_2$ and $L_3$ be 
distributed as in $\lambda_c$, $\lambda_q$, and $\lambda_s$, respectively. Conditional that the 
pair  $(L_1, L_2)$ is rescuable and $L_3$ is  good w.r.t $L_1$, then $L_1$ and $L_3$ are 
identically distributed. 
\end{lemma}

\noindent
For the proof of Lemma \ref{lemma:RescuabelVsGood} see in Section \ref{sec:lemma:RescuabelVsGood}.

\subsection{The coupling}

The coupling works inductively. At each step it considers two consecutive levels of $T$. 
Here, we describe how does it work for the two levels below the root. The coupling for the
rest of the tree will be immediate.

We need to use the some auxiliary random variables defined w.r.t. $X, Y$. We let $L_{X}, L_Y \in [k]^d$ 
be ordered lists which contain the colours that are assigned to the children of the root $r$  under 
the colour assignments $X$ and $Y$, respectively
\footnote{There is a bijection between the elements $L_X$ and the colour assignments of $X$ at the
children of the root $r$. The same holds for $Y$ and $L_Y$.}.   
Additionally, for every $i\in [d]$ we let $L^i_{X}$ (and $L^i_Y$) be the corresponding lists of the 
colour assignments of the children of the vertex that is going to be assigned the colour $L_{X}(i)$
(and $L_Y(i)$), e.g.  see  Figure \ref{fig:Lists-LY}.

\begin{figure}
\begin{minipage}{0.5\textwidth}
	\centering
		\includegraphics[height=2.1cm]{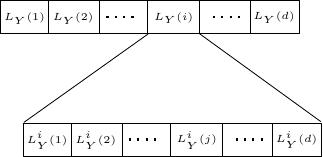}
		\caption{The lists $L_Y$ and $L^i_Y$.}
	\label{fig:Lists-LY}
\end{minipage}
\begin{minipage}{0.5\textwidth}
	\centering
		\includegraphics[height=2.5cm]{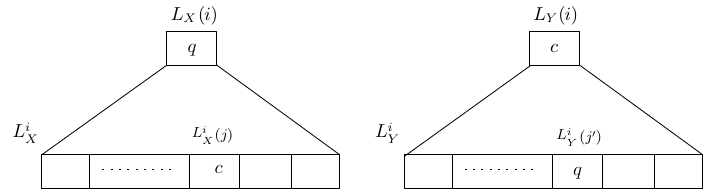}
		\caption{A ``Bad'' pair of lists.}
	\label{fig:Lists-BadLY}
\end{minipage}
\end{figure}

Essentially, the list $L_X$ specifies the colours that are assigned to the children of 
the root by $X$ but without providing exactly the information about which vertex takes 
which colour. The same holds for the other lists $L^i_X$, $L_Y$ and $L^i_Y$,
for every $i\in [d]$. We couple the colour assignments of $X,Y$ on the vertices
at levels 1 and 2 of $T$ by using these lists.  I.e. we couple the entries 
of the lists, first, and then we obtain the assignments of $X,Y$. There we need to 
use the following

\begin{remark}\label{remark:ListVsAssignments}
Given $L_X$, the colour assignments of $X$ to  the children of $r$ can be obtained as follows: 
Take $\pi$, a random permutation of the elements in $\{1, \ldots, d\}$. Then, for the $i$-th
child of $r$ set $X(i)=L_X (\pi(i))$. Given $L^i_X$s we obtain the colourings of the grandchildren
of $r$ in an analogous way.
\end{remark}

\noindent
We use the notions of the ``bad'' or the ``rescuable'' pair for every  $(L^i_X, L^i_Y)$ 
such that $L_X(i)=q$ and $L_Y(i)=c$, with $i\in [d]$. 
That is, we consider a pair  $(L^i_X,L^i_Y)$ to be bad (or not) only if $L_X(i)=q$ and $L_Y(i)=c$. 
For $(L^i_X,L^i_Y)$ ``badness'' is determined w.r.t  the colours $c,q$. E.g. see Figure 
\ref{fig:Lists-BadLY}. There $L^i_X$ does not contain $q$ due to the fact that $L_X(i)=q$ but it 
contains $c$, i.e. $L^i_X(j)=c$. On the other hand, $L^i_Y$ does not contain $c$ due to the fact that $L_Y(i)=c$ but 
it contains  $q$, i.e. $L^i_Y(j')=q$.

The coupling works in three phases. In the first two it focuses on the list of
colours. It considers $X$ and $Y$ only in the last phase.

In {\bf  Phase 1} only a certain part of information about $L_X$, $L_Y$, $L^j_X$ and
$L^j_Y$, for $j\in [d]$, is revealed. That is, we reveal  ``bad'' and ``rescuable'' pairs as well as which 
lists are ``special''.  Observe that the lists $L_X$, $L_Y$ are distributed as in $\lambda_c$ 
and $\lambda_q$, respectively.  Also, given that $L_X(i)=c'$ (or $L_Y(i)=c'$) for some $c'\in [k]$,
then $L^i_X$ (or $L^i_Y$) is distributed as in $\lambda_{c'}$. Phase 1 is as follows.
\\  \vspace{-.1cm}

\noindent
{\bf Phase 1:} \\ \vspace{-.8cm} \\
\rule{\textwidth}{1pt} \\ 
\vspace{-.7cm}
\begin{enumerate} 
\item 
Reveal only for which $i\in [d]$ we have $L_X(i), L_Y(i)\in \{c,q\}$. Couple the choices of $L_X$ and $L_Y$ 
such that $L_X(i)=q$ if and only if $L_Y(i)=c$.

\item For each $i$ such that $L_X(i)=q$ and $L_Y(i)=c$  reveal whether $(L^i_X, L^i_Y)$ 
is ``bad'' or not \footnote{Here, we only ask if $L^i_X$ and $L^i_Y$ contain $c$ and $q$, respectively.
It is trivial that $Pr[c\in L^i_X]=Pr[q\in L^i_Y]$. }. 

\item If $(L^i_X,L^i_y)$ is bad reveal whether it is ``rescuable''. The coupling is so that 
the colours in $[k]\backslash \{c,q\}$ are chosen independently from the two lists.

\item If the number of rescuable pairs is $l$, partition the set of non-bad pairs 
$(L^j_X, L^j_Y)$  into $l$ parts which are as equal sized as possible.  Each rescuable pair 
is {\em associated} to exactly one part in the partition.

\item For each non-bad pair $(L^j_X, L^j_Y)$ that is associated to 
the rescuable pair $(L^i_X, L^i_Y)$ do the following: 
Reveal if $(L^j_X, L^j_Y)$  consists of special lists, i.e. $L^j_X$ and $L^j_Y$ are 
special w.r.t $L^i_Y$ and $L^i_X$, respectively. We use coupling such that either both lists 
in the pair are special or both are not.
\end{enumerate}
\vspace{-.5cm}
\rule{\textwidth}{1pt} \\ \vspace{-.1cm}

\noindent
We should recognize the bad pairs as the potential sources of disagreements in the coupling.
%It is impossible to couple the lists in a bad pair identically. 
Our attempt is to eliminate the disagreements caused by the rescuable pairs only \footnote{For the values of $k$ we
consider it is highly unlikely that a bad pair is non-rescuable.}. This eliminations
uses Lemma \ref{lemma:RescuabelVsGood} as follows:
Consider the rescuable pair $(L^i_X,L^i_Y)$. We let $A_i$ be the set of indices 
such that if $j\in A_i$ then the pair $(L^j_X,L^j_Y)$ is associated to the bad pair  
$(L^i_X,L^i_Y)$ in step 4. Assume that the pair $(L^j_X,L^j_Y)$ is ``$i$-good'', i.e. 
$L^j_X$ is good w.r.t $L^i_Y$  and $L^j_Y$ is good w.r.t. $L^i_X$. Then, Lemma 
\ref{lemma:RescuabelVsGood} implies that $L^j_X$ and $L^i_Y$ are identically 
distributed. The same holds for $L^j_Y$ and $L^i_X$.
In this case,  when we reveal all the information of the lists (which will be done 
in a subsequent phase)  we couple $L^j_X$ with  $L^i_Y$  and $L^j_Y$ and $L^i_X$. 
Clearly this eliminates all the potential disagreements generated by the rescuable pair 
$(L^i_X,L^i_Y)$. %{\color{red}Say something extra about the colour assignments}

\begin{remark}
For technical reasons which will become apparent soon, we do not reveal which pairs in $A_i$
are $i$-good. We only reveal if the pair $(L^j_X, L^j_Y)$, for $j\in A_i$, is $i$-special, i.e. 
$L^j_X$ and $L^j_Y$ are special w.r.t $L^i_Y$ and $L^i_X$, respectively.
\end{remark}

\noindent
In {\bf Phase 2}, we construct a mapping $f:[d]\to[d]$  with the following property: 
If $f(i)=j$, then when we reveal the full information about the lists we couple maximally 
$L_X(i)$ with $L_Y(j)$ and $L^i_X$ with $L^j_Y$. The mapping  $f$ is constructed so as to 
minimize the number of disagreements between the lists $L^i_X$ and $L^j_Y$.
In particular we have the following situation in mind. It is desirable that for each rescuable 
pair $(L^i_X,L^i_Y)$ to find an $i$-good pair among the $i$-special pairs in $A_i$. Once we have 
such a pair, e.g. $(L^j_X,L^j_Y)$ for some $j\in A_i$, we set $f(i)=j$ and $f(j)=i$. 

The next step of the coupling reveals which $i$-special pairs $(L^j_X,L^j_Y)$ with $j\in A_i$ 
are also $i$-good.
In order to reveal whether an $i$-special pair $(L^j_X,L^j_Y)$, for $j\in A_i$, is  $i$-good 
we should couple the lists such that $L^j_X \neq L^j_Y$. This pair is $i$-good with probability 
$1/2$. With the remaining  probability it is not and  the lists in the pair  
$(L^j_X,L^j_Y)$ cannot be coupled identically.

\begin{remark}
For the $i$-special pair $(L^j_X,L^j_Y)$, we reveal whether ``$c\in L^j_X$ and $q \notin L^j_X$'' 
or ``$c\notin L^j_X$ and $q\in L^j_X$''. E.g. assume that we have ``$c\in L^j_X$ and $q\notin L^j_X$'',
then the coupling should decide the opposite for $L^j_Y$, i.e. ``$c\notin L^j_Y$ and $q\in L^j_Y$''. 
Revealing the lists in such a way it always holds $L^j_X\neq L^j_Y$
\end{remark}

\noindent
Of course  there is always the option of coupling an $i$-special pair identically. 
But then it is impossible to generate an $i$-good pair. The $i$-special pair $(L^j_X,L^j_Y)$ 
which is coupled so as to generate an $i$-good pair but it failed to do so is called  
{\bf $i$-fail} (see example in Figure \ref{fig:i-matching}, the upper pair is $i$-fail). 
It is straightforward, now, that as we search for an $i$-good pair it is possible that we 
generate extra potential sources of  disagreements.  To this end we use the following lemma.

\begin{lemma}\label{lemma:GoodVsFail}
Assume that the $i$-special pairs $(L^t_X,L^t_Y)$ and $(L^s_X,L^s_Y)$ with $s,t\in A_i$ are revealed 
and $(L^t_X,L^t_Y)$ is $i$-good while $(L^s_X, L^s_Y)$  is $i$-fail. Then,  $L^t_X$ is identically 
distributed to $L^s_Y$ and $L^t_Y$ is identically distributed to $L^s_X$.
\end{lemma}

\noindent
For a proof of Lemma \ref{lemma:GoodVsFail} see in Section \ref{sec:lemma:GoodVsFail}.
Figure \ref{fig:i-matching} gives a schematic representation of what is stated 
in Lemma \ref{lemma:GoodVsFail}. The arrows show the pairs of lists that are identically distributed.  

Lemma \ref{lemma:GoodVsFail} suggests that $i$-good pairs can be used to eliminate
the potential disagreements generated by $i$-fails. Thus, in the case we generate $i$-fails  
we (try) to reveal some extra $i$-good pairs. In particular, we work as follows: 
\\ \vspace{-.3cm}

\noindent
{\bf Phase 2.} %\hspace*{1cm} /*List Association*/ 
\\ \vspace{-.8cm} \\
\rule{\textwidth}{1pt} \\ 
\vspace{-.4cm}

\noindent
For each rescuable pair $(L^{i}_X,L^i_Y)$ do the following:
\vspace{-.1cm}
\begin{enumerate}
\item Reveal, sequentially, whether each $i$-special pair in $A_i$  is $i$-good or $i$-fail until 
either of the following
two happens:
\begin{itemize}
\item the number of $i$-good pairs exceeds the number of $i$-fails by one, \vspace{-.1cm}
\item there are no other $i$-special pairs in $A_i$ to reveal.
\end{itemize}
\vspace{-.2cm}
\item The remaining unrevealed $i$-special pairs, if any, are coupled by using
identity coupling. 

\item If there is an $i$-good pair $(L^j_X,L^j_Y)$  ``match'' it with the rescuable pair 
$(L^i_X, L^i_Y)$, i.e. set $f(i)=j$ and $f(j)=i$.

\item Each of the remaining $i$-good pairs $(L^j_X,L^j_Y)$ is
matched to one $i$-fail pair $(L^s_X,L^s_Y)$, i.e. set $f(j)=s$ and $f(s)=j$.
No $i$-fail is matched to more than one  $i$-good pairs and vice versa
\footnote{This implies that the mapping $f$ is a bijection.}.

\item For each $j\in A_i$ such that $(L^j_X,L^j_Y)$ is not matched
to some other pair, match it to itself, i.e. set $f(j)=j$.\vspace{-.1cm}
\end{enumerate}
\noindent
\rule{\textwidth}{1pt} \\

\begin{figure}
\begin{minipage}{\textwidth}
	\centering
		\includegraphics[height=5cm]{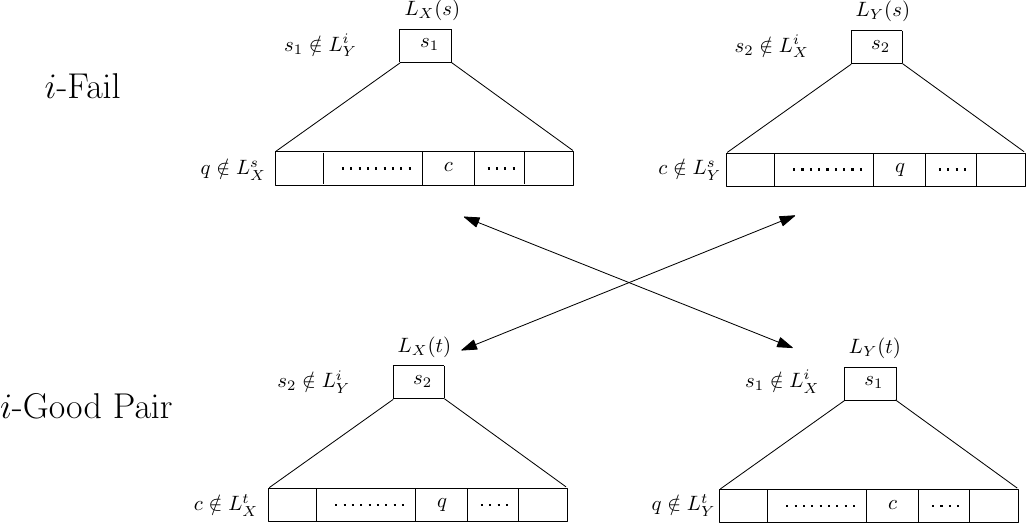}
		\caption{Matching between $i$-fail and $i$-good pairs. The coupling between $L_X(s)$-$L_Y(t)$ and
		$L_X(t)$-$L_Y(s)$ is done after the good/fail revelations.}
	\label{fig:i-matching}
\end{minipage}
\end{figure}

\noindent
Ideally, Phase 2 generates a number of $i$-good pairs which exceed the number of $i$-fails
by one. If this is the case, $f$ specifies pairs whose coupling generates no disagreement.
That is, the rescuable pair $(L^i_X,L^i_Y)$ and the $i$-fails are going to be coupled with 
an $i$-good pair each. Then, due to Lemma \ref{lemma:RescuabelVsGood} and Lemma \ref{lemma:GoodVsFail}
no disagreement is going to be generated.
Of course, it is possible that the number of the $i$-good pairs is not sufficiently large.
Then, we end up with some $i$-fails which  cannot be matched  with any $i$-good pair (possibly 
with the rescuable pair $(L^i_X,L^i_Y)$ as well). These pairs are matched to themselves and
some disagreements are going to appear in the full revelation. However, we show that the
expected number of disagreements vanishes as long as $k\geq (3+\epsilon)d/\ln d$.

We now, proceed with {\bf Phase 3}. There we  reveal the  full information about the lists 
by coupling the pairs  as specified  by $f$. Given the full information for the lists we reveal the 
assignments of $X$, $Y$ for the (grand)children of $r$. 
Note that if $f(i)=j$, then the child of $r$ that gets $L_X(i)$ under
$X$ will get $L_Y(j)$ under $Y$.  Additionally, the grand child of $r$
that is assigned the colour $L^i_X(t)$ under $X$ is going to take the
colour $L^j_Y(t)$ under $Y$. 
\\ \vspace{-.3cm}

\noindent
{\bf Phase 3:} \hspace*{1cm}  \\ \vspace{-.7cm} \\
\rule{\textwidth}{1pt} \\ 
\vspace{-.7cm}

\begin{enumerate}
\item  For every $s,t$ such that $f(s)=t$, couple optimally $L_X(s)$
with $L_Y(t)$ as well as $L^s_X$ with $L^t_Y$.

\item Reveal which element of the list $L_X$ is assigned to which
child of $r$ and which element of  $L^j_X$  goes to which grandchild of
$r$, as Remark \ref{remark:ListVsAssignments} specifies.

\item  Assuming that $v$, child of $r$, is such that $X(v)$ is set
$L_X(s)$, then we set $Y(v)$ equal to $L_Y(t)$, where $t=f(s)$.
Also, for $u$, child of $v$, such that $X(u)$ set $L^s_X(j)$ 
we set $Y(u)$  equal to $L^t_Y(j)$.
\end{enumerate}
\vspace{-.5cm}
\rule{\textwidth}{1pt} \\

\noindent
{%\color{red}
Applying the coupling inductively, i.e. for the grandchildren of the root and so on, 
at the end we get the full colourings $X$ and $Y$.
%
%For completeness, in Section \ref{sec:FullCoupling} we provide the full pseudocode of 
%coupling as a recursive procedure.
}
A very basic result is the following theorem.

\begin{theorem}\label{thrm:ValidOptimal}
For $c,q\in [k]$, assume that in the above coupling it holds
$X(r)=c$ and $Y(r)=q$, where $r$ is the root vertex of $T$. 
Then at the end of the coupling, $X$ and $Y$ are distributed
as in $\mu(\cdot|X(r)=c)$ and $\mu(\cdot|Y(r)=q)$, respectively.
\end{theorem}

\begin{proof}
Theorem follows by noting that for every list, conditional on the
information that is already known to us, we reveal some information
by using the appropriate distribution.
\end{proof}

\noindent
Furthermore, from the description of the coupling the following
corollary is direct.

\begin{corollary}\label{cor:SourcesOfDis}
The disagreements in the coupling have three different sources:
\begin{enumerate}
\item Pairs of bad lists which are not rescuable.
\item Pairs of rescuable lists for which it was impossible to find a good pair.
\item Pairs of $i$-fail lists, for some $i$, which are not matched
to an $i$-good pair.
\end{enumerate}
\end{corollary}

\begin{proposition}\label{prop:Expect-Dis}
Consider the above coupling between $X$ and $Y$ and let ${\cal W}_l$
be the number of  vertices $u$ at level $l$ such that $X(u)\neq
Y(u)$. For fixed $\epsilon>0$, sufficiently large $d$, 
$k=(1+\epsilon)\frac{d}{\ln d}$ and every even integer $l>0$ it holds that 
\begin{displaymath}
E[{\cal W}_l]\leq \left(d^{-0.1\frac{\epsilon-2}{\epsilon+1}}\right)^{l/2}.
\end{displaymath}
\end{proposition}

\noindent
Using Proposition \ref{prop:Expect-Dis} it is direct to see that our 
combinatorial construction implies the following theorem.

\begin{theorem}\label{thrm:MyNonReconstruction}
For fixed $\epsilon>0$ and sufficiently large $d$, the following is true 
for the $k$-colouring model on a $d$-ary tree $T$: If $k=(3+\epsilon)d/\ln d$, 
then the model is non-reconstructible.
\end{theorem}
\begin{proof}
Take $k=(3+\epsilon)d/\ln d$. Let  $X$ and $Y$ be distributed as in $\mu(\cdot|X(r)=c)$ and $\mu(\cdot|Y(r)=q)$, 
respectively, while their joint distribution is specified by the coupling we presented. Let the set $L_h$ 
contain all the vertices of $T$ at level $h$. We take $h$ to be even.
By Coupling Lemma we have

\begin{equation}\label{eq:CouplingLemma}
||\mu(\cdot|X(r)=c)-\mu(\cdot|Y(r)=q)||_{L_{h}}\leq Pr[X(L_h)\neq Y(L_h)].
\end{equation}
Let ${\cal W}_h$ be the number of vertices $u\in L_h$ such that
$X(u)\neq Y(u)$. It holds that
\begin{eqnarray}
Pr[X(L_h)\neq Y(L_h)]&=& Pr[{\cal W}_h>0]
\leq E[{\cal W}_h] \hspace{1.5cm} \mbox{[by Markov's inequality]} \nonumber\\
&\leq& \left(d^{-0.1\frac{\epsilon}{\epsilon+2}}\right)^{h/2}.
\hspace{2.85cm} \mbox{[from Proposition \ref{prop:Expect-Dis}]} \label{eq:Proposition1Usage}
\end{eqnarray}
The theorem follows by combining (\ref{eq:CouplingLemma}) and (\ref{eq:Proposition1Usage}).
\end{proof}

\section{Proof of Proposition \ref{prop:Expect-Dis}}\label{sec:prop:Expect-Dis}

\noindent
Consider in the coupling two vertices $v,w\in T$  at the same level $l$, where $l$ is even. 
Consider, also, the colourings $X(v)$, $Y(v)$ and $X(w)$ and $Y(w)$ while w.l.o.g
assume that $X(v)\neq Y(v)$ and $X(w)\neq Y(w)$.  Clearly,  whether the descendants of $v$  
disagree or not does not dependent on what happens at the descendants of $w$ and vice
versa. This observation yields the following: In the coupling, for each vertex  $v\in T$, 
let ${\cal D}_v$ be the number of disagreements two levels below $v$.  It holds that
\begin{displaymath}
E[{\cal W}_l|{\cal W}_{l-2}]=E[{\cal D}_v]\cdot {\cal W}_{l-2} \qquad \mbox{for even $l>0$.}
\end{displaymath}
Taking the average from both sides and working out the recursion
we get that
\begin{displaymath}
E[{\cal W}_l]=\left(E[{\cal D}_v]\right)^{l/2}.
\end{displaymath}
The proposition will follow by bounding appropriately $E[{\cal D}_v]$.  To this end, 
we need to bound the number of disagreements that are generated by each of the three 
sources of disagreement specified in Corollary \ref{cor:SourcesOfDis}. It, always, holds 
that $D_v\leq d^2$, since $T$ is a $d$-ary tree.

Consider the following quantities related to the vertex $v$: Let $\beta_v$ denote the number 
of bad pairs of lists two levels below $v$. Let $\delta_k$ be the probability for a bad pair to be 
rescuable, for a given number of colours $k$. Finally, given some rescuable pair  $(L^j_X,L^j_Y)$ 
let $h^j_v$ be the number of $j$-special lists in the associated partition. 
Let the event $\mathbb{A}$ denote that at least  one of the following three occurs
\begin{enumerate}
\item $\beta_v\geq 100\ln d$.
\item There is at least one bad pair which is not in a rescuable pair.
\item There is a rescuable pair $(L^j_X,L^j_Y)$ that is associated
to a partition with less than $d^{\frac45\frac{\epsilon-2}{1+\epsilon}}$
$j$-special lists.
\end{enumerate}

\noindent
It is direct to get that
\begin{equation}\label{eq:TargeEDvBound}
E[D_v]\leq d^2Pr[{\mathbb A}]+E[D_v|{\mathbb A}^c],
\end{equation}
where we  use  the rather crude overestimate that conditional on the event
$\mathbb{A}$ occurs all the $d^2$ descendants of $v$ are disagreeing. 
It suffices
to bound appropriately $Pr[{\mathbb A}]$ and $E[D_v|{\mathbb A}^c]$.
To this end, we use the following propositions.

\begin{proposition}\label{prop:GoodCaseDisagreement}
For $k=(1+\epsilon)d/\ln d$ and for sufficiently large $d$, we have that
\begin{displaymath}
E[D_{v}|\mathbb{A}^c]\leq d^{-0.102\frac{\epsilon-2}{\epsilon+1}}.
\end{displaymath}
\end{proposition}

\begin{proposition}\label{prop:BadCaseProb}
For $k=(1+\epsilon)d/\ln d$ and for sufficiently large $d$, we have that
$$ Pr[\mathbb{A}]\leq 5 d^{-250}. $$
\end{proposition}

\noindent
Plugging into (\ref{eq:TargeEDvBound}) the bounds from Proposition \ref{prop:GoodCaseDisagreement}, Proposition \ref{prop:BadCaseProb}  we get that
\begin{displaymath}
E[D_v]\leq d^{-0.1\frac{\epsilon-2}{\epsilon+1}}.
\end{displaymath}
The proposition follows.

\subsection{Proof of Proposition \ref{prop:GoodCaseDisagreement}}
\label{sec:prop:GoodCaseDisagreement}

Since we have conditioned on $\mathbb{A}^c$, we have that 
A) $\beta_v$, the number of bad lists, is less than $100\ln d$, 
B)  all the bad lists are rescuable and 
C) every rescuable pair $(L^i_X,L^i_Y)$ is associated to a partition which contains 
at least $d^{\frac{4}{5}\frac{\epsilon-1} {1+\epsilon}}$ $i$-special lists. 
From (A) and (B), we deduce that the number of rescuable pairs is equal to $\beta_v$.

In this setting, consider the rescuable pair $(L^i_X, L^i_Y)$.  We remind the reader that 
during the second phase of the coupling, in the partition associated to $(L^i_X, L^i_Y)$, we 
reveal which of the $i$-special pairs are $i$-good or not, i.e. during the steps 1 and 2. During 
these revelations it is possible that we introduce pairs which are $i$-fails which my end up being coupled
together (due to lack of $i$-good pairs). Let $\Delta_i$ be the indices of these $i$-fails.
%i.e. which it was  impossible to find an $i$-good pair. 

We remind the reader that we denote with $A_i$ the set of indices of the pairs %in the partition 
that  are associated to the rescuable pair $(L^i_X, L^i_Y)$.

Consider $(L^t_X,L^t_Y)$ for some $t\in \Delta_i$. We can couple $L_X(t)$, $L_Y(t)$ such that $L_X(t)=L_Y(t)$. 
Also, it holds that  $c\in L^t_X$ and $q\notin L^t_X$ while $q\in L^t_Y$ and $c \notin  L^t_Y$. 
Given that $L_X(t) =L_Y(t)$, all the colours in $[k]\backslash\{c,q,L_X(t)\}$ are symmetric for both 
$L^t_X$  and $L^t_Y$. Clearly, we can couple $L^t_X$ and $L^t_Y$, such that if $L^t_X(s)=c$ then $L^t_Y(s)=q$ 
while if $L^t_X(s)\neq c$, then $L^t_X(s)=L^t_Y(s)$ for any $s\in [d]$.

Let $Z_t$ be the number of disagreements that are generated by the 
coupling of the pair $(L^t_X,L^t_Y)$ with $t\in \Delta_i$.  Also, 
let $Q_i=\sum_{t\in \Delta_i}Z_t$.  It holds that
\begin{eqnarray}
E[Q_i|\mathbb A^c]=E[|\Delta_i||\mathbb A^c]\cdot E[Z_j|\mathbb A^c].
\label{eq:Q_iRelation}
\end{eqnarray}

\noindent
Apart from the pairs in $\Delta_i$, it is possible that the lists in the rescuable pair 
$(L^i_X,L^i_Y)$ are coupled together. This happens when there is no $i$-good pair among 
the $i$-specials. The probability  of having no $i$-good pairs at most $2^{-d^{\frac{4}{5}\frac{\epsilon-2}{\epsilon+1}}}$,
as every special pair is $i$-good with probability $1/2$ and we have at least 
$d^{\frac{4}{5}\frac{\epsilon-2}{\epsilon+1}}$ $i$-special pairs.  
Let $W_i$ be the number of disagreements that are generated by 
the rescuable pair. It holds that 
\begin{eqnarray}
 E[W_i|\mathbb A^c]\leq dPr[\textrm{No $i$-good pair in $A_i$}|{\mathbb A}^c]\leq 2^{-d^{\frac{3}{5}\frac{\epsilon-2}{\epsilon+1}}}.
\label{eq:EWiBound}
\end{eqnarray}

\noindent
Conditional on $\mathbb{A}^c$,  $D_v$ is the sum of disagreements generated by the 
rescuable pairs and the  $i$-fails, for various $i$. By the linearity of expectation we get that
\begin{eqnarray}
E[D_v|\mathbb A^c]&\leq& (100\ln d)\left( E[W_i|\mathbb A^c]+ E[Q_i|\mathbb A^c]\right) \hspace*{2.25cm}  \mbox{[as $\mathbb{A}^c$
assumes that $\beta_v <100\ln d$]}
\nonumber \\
&\leq& 2^{-d^{\frac{1}{2}\frac{\epsilon-2}{\epsilon+1}}}+ (100\ln d)\cdot E[|\Delta_i||\mathbb A^c]\cdot E[Z_j|\mathbb A^c]. 
\hspace{.51cm} \mbox{[from (\ref{eq:EWiBound}) and (\ref{eq:Q_iRelation})]}\label{eq:DvDeltaIZJ}
\end{eqnarray}
The proposition will follow by bounding appropriately $E[|\Delta_i||\mathbb A^c]$
and $E[Z_j|\mathbb A^c]$.

As far as $E[Z_j|\mathbb A^c]$ is concerned we have the following:
For any $t\in \Delta_i$, the lists  $(L^t_X,L^t_Y)$  the number of disagreements is exactly the number 
of occurrences of $c$ in $L^t_X$. Conditional on $\mathbb A^c$, the number of entries in $L^t_X$ with 
colour $c$ is binomially distributed with parameters $d,1/(k-1)$, conditional that it is positive. 
It follows that 
\begin{eqnarray}
E[Z_j|\mathbb A^c]&=&\sum_{s=0}^d s\cdot Pr[c\;\textrm{appears $s$ times in $L^t_X$}|c\; \textrm{appears at least once in $L^t_X$}]\nonumber \\
&=& \left({1-\left(1-\frac{1}{k-1}\right)^d}\right)^{-1}\sum_{s=1}^d s\cdot {d \choose s} \left(\frac{1}{k-1}\right)^s\left(1-\frac{1}{k-1}\right)^{d-s}\nonumber \\
&\leq& 2 \frac{d}{k-1} \hspace{5cm}  \mbox{[since $1-\left(1-\frac{1}{k-1}\right)^d>1/2$]}
\nonumber\\
&\leq& 2 \ln d. \hspace{6cm} \mbox{[since $k=(1+\epsilon)d/\ln d$]}\label{eq:Z_jBound4Dv}
\end{eqnarray}

\noindent
As far as $E[|\Delta_i||\mathbb A^c]$ is concerned, we work as follows: Let $S_i$ be the set of indices
of all the $i$-special pairs in $A_i$ as well as of the rescuable pair $(L^i_X,L^i_Y)$. W.l.o.g. assume
that $i=1$ while the indices of the $i$-special pairs in $S_i$ are from $2$ to $|S_i|$. 
Let the 0-1 matrix ${\cal S}= |S_i|\times 2$ be defined as follows: ${\cal S}(1,t)=1$, if $c\in L^t_X$ and $q
\notin L^t_X$,  otherwise, i.e.  $c\notin L^t_X$ and $q \in L^t_X$, ${\cal S}(1,t)=0$. Similarly, 
${\cal S}(2,t)=1$ if $c\notin L^t_Y$ and $q\in L^t_Y$, otherwise ${\cal S}(2,t)=0$. 
If the $i$-special pair $(L^t_X,L^t_Y)$ is $i$-good, the it holds 
that $({\cal S}(1,t), {\cal S}(2,t)) =(0,1)$, otherwise, i.e. the pair is $i$-fail, 
then $({\cal S}(1,t), {\cal S}(2,t))=(1,0)$.

\begin{remark}\label{remark:CouplingString} 
\em
The second phase of the coupling specifies how  ${\cal S}(1,j)$ and ${\cal S}(2,j)$
are correlated with each other, i.e. the following holds: if 
$\sum_{j=1}^{i-1}\left({\cal S}(1,j)-{\cal S}(2,j)\right)> 0$, then 
${\cal S}(1,i)$ and ${\cal S}(2,i)$  get complementary values. 
Otherwise, i.e. if $\sum_{j=1}^{i-1}\left({\cal S}(1,j)-{\cal S}(2,j)\right)=0$, they
are identical.
\end{remark}

\noindent
Since we have assumed that the values in $({\cal S}(1,1),(2,1))$ are specified by the rescuable
pair $(L^i_X,L^i_Y)$, by definition, it holds that $({\cal S}(1,1),(2,1))=(1,0)$. Furthermore, 
for each $t= 2\ldots |S_i|$ and as long as $\sum_{j=1}^{t-1}\left({\cal S}(1,j)-{\cal S}(2,j)\right)> 0$ 
we have
\begin{displaymath}
({\cal S}(1,t),{\cal S}(2,t))=\left\{
\begin{array}{lcl}
(1,0) &\quad & \textrm{with probability } 1/2\\
(0,1) &\quad & \textrm{with probability } 1/2.\\
\end{array}
\right.
\end{displaymath}

\noindent
For the matrix ${\cal S}$ we have the following lemma.

\begin{lemma}\label{lemma:DeltaIVsBinaryString}
Let $N$ be the number of columns of the matrix ${\cal S}$.
%The cardinality of $\Delta_i$ and ${\cal S}$ are related
%as follows:
It holds that
\begin{displaymath}
|\Delta_i|\leq \sum_{t=1}^{N} {\cal S}(1,t)-{\cal S}(2,t).
\end{displaymath}
%where $N$ is the number of columns of the matrix ${\cal S}$.
\end{lemma}
\begin{proof}
First notice that $S(1,1)-S(2,1)=1$. The coupling during the second phase assigns complementary
values to each pair $S(1,t)$, $S(2,t)$ as long as $R_t=\sum_{i=1}^{t-1}[S(1,t)-S(2,t)]>0$. Once
$R_t=0$ it sets $S(1,t)=S(2,t)$, i.e. $R_t$ remains zero for the rest values of $t$.

Let $T$ be the maximum $t$ such that $S(1,t)\neq S(2,t)$. It suffices to show that 
$$|\Delta_i|\leq \sum_{t=1}^{T} {\cal S}(1,t)-{\cal S}(2,t).$$ For $t<T$, the fact that ${\cal S}(1,t)=1$
(and consequently ${\cal S}(2,t)=0$) suggests that we have revealed an $i$-fail. On the other
hand, if ${\cal S}(1,t)=0$ (and consequently ${\cal S}(2,t)=1$), then it suggests that it has
been revealed an $i$-good pair. This observation implies that the sum $\sum_{t=1}^T{\cal S}(1,t)$
is equal to the number of $i$-fails we have revealed plus one, while $\sum_{t=1}^T{\cal S}(2,t)$ is equal to
the number of $i$-good pairs.

Since we can match an $i$-fail with an $i$-good pair to avoid generating disagreements, the number 
of pairs which do not admit identical coupling, i.e. the $i$-fail and possibly the rescuable pair, 
is equal to 
$$
\sum_{t=1}^T{\cal S}(1,t)-{\cal S}(2,t)=\sum_{t=1}^N{\cal S}(1,t)-{\cal S}(2,t).
$$ 
The lemma follows.
\end{proof}

\begin{proposition}\label{prop:PlantingBits}
Let $N$ be the number of columns of ${\cal S}$. Then for sufficiently large
$N$ it holds that
\begin{equation}%\label{eq:PlantedStringExpct}
E\left[\sum_{j=1}^N\left({\cal S}(1,j)- {\cal S}(2,j)\right)\right]\leq 
\left(\frac{2.3}{\pi}\right)^{0.43 \ln N}. \nonumber
\end{equation}
\end{proposition}
For a proof of Proposition \ref{prop:PlantingBits} see in  Section \ref{sec:planting}.

Using Proposition \ref{prop:PlantingBits} and Lemma \ref{lemma:DeltaIVsBinaryString}
and the assumption that the number of $i$-special pairs in $A_i$ is at least 
$d^{\frac{4}{5}\frac{\epsilon-2}{\epsilon+1}}$, we get 
\begin{equation}\label{eq:ExpDeltaIBound}
E[|\Delta_i||\mathbb A^c]\leq \left(\frac{2.3}{\pi}\right)^{0.43\frac{4(\epsilon-2)}{5(\epsilon+1)}\ln d}\leq d^{-0.344\frac{\epsilon-2}{\epsilon+1}\ln\left(\frac{\pi}{2.3}\right)}
\leq d^{-0.107\frac{\epsilon-2}{\epsilon+1}}.
\end{equation}
Plugging the inequalities (\ref{eq:Z_jBound4Dv}) and (\ref{eq:ExpDeltaIBound}) into
(\ref{eq:DvDeltaIZJ})  we get that
\begin{displaymath}
E[D_v|\mathbb A^c]\leq 200(\ln d)^2d^{-0.107\frac{\epsilon-2}{\epsilon+1}}.
\end{displaymath}
The proposition follows.

\subsection{Proof of Proposition \ref{prop:BadCaseProb}}
\label{sec:prop:BadCaseProb}
For the  quantities, $\beta_v, \delta_k$ and $h_v$  we defined in Section \ref{sec:prop:Expect-Dis} 
we have the following proposition.

\begin{proposition}\label{prop:BoundsOnNumberLists}
For $k=(1+\epsilon)d/\ln d$ the following are true:
\begin{equation}\label{eq:BetaProbBound}
Pr\left[\beta_v\geq (1+x)\frac{d}{k-1}\right]\leq 
d^{ \left(-\frac{3\phi(x)}{4(1+\epsilon)}\right)},
\end{equation}
where $\phi(x)=(1+x)\ln(1+x)-x$, for real $x>0$. Also, it holds
that
\begin{equation}\label{eq:RescProbBound}
\delta_k\geq 
1-2\exp \left(
-\frac{(1+\epsilon)}{24\ln d}d^{\frac{\epsilon}{1+\epsilon}}
\right).
\end{equation}
Finally, for any $c>0$ it holds that
\begin{equation}\label{eq:hProbBound}
Pr\left[h_v\leq \frac{d^{\frac{\epsilon-2}{1+\epsilon}}}{16c\ln d}
\left| \right. \beta_v\leq c\ln d\right]
\leq \exp\left(-\frac{d^{\frac{\epsilon-2}{1+\epsilon}}}{64c\ln d}
\right).
\end{equation}
\end{proposition}
The proof of Proposition \ref{prop:BoundsOnNumberLists} appears in Section \ref{sec:prop:BoundsOnNumberLists}.

Let the events $E_1=$``$\beta_v\geq 100 \ln d$'', $E_2=$``there is at least one bad pair of lists 
which is not rescuable'' and $E_3=$``there is a pair rescuable lists $(L^j_X,L^j_Y)$ that is associated
to a partition with less than $d^{\frac45\frac{\epsilon-2}{1+\epsilon}}$ $j$-special pairs''.
From a simple union bound we get that 
\begin{eqnarray}\label{eq:AEiRelation}
Pr[\mathbb{A}]=Pr\left[\bigcup_{i=1}^3E_i\right]\leq \sum_{i=1}^3Pr[E_i].
\end{eqnarray}
The proposition will follow by bounding appropriately the probability terms
$Pr[E_1]$,$Pr[E_2]$ and $Pr[E_3]$. As far as $Pr[E_1]$ is regarded it holds that
\begin{eqnarray}
Pr[E_1]\leq Pr\left[\beta_v>(1+x_0)\frac{d}{k-1}\right], \label{eq:PE1+PropBound}
\end{eqnarray}
where $1+x_0=98(1+\epsilon)$ . The above inequality holds since
$\frac{d}{k-1}\leq \frac{\ln d}{(1+\epsilon)}+\frac{2\ln^2 d}{d}$.

We use Proposition \ref{prop:BoundsOnNumberLists}, (i.e. (\ref{eq:BetaProbBound}))
to bound the r.h.s of (\ref{eq:PE1+PropBound}). In particular, for
$x_0=98(1+\epsilon)-1$ it holds that $\phi(x_0)\geq 343(1+\epsilon)+98(1+\epsilon)\ln(1+\epsilon)$.
Then, from (\ref{eq:PE1+PropBound}) we get that 
\begin{equation}\label{eq:PE1+PropBoundFound}
Pr[E_1]\leq d^{-250}.
\end{equation}

\noindent
As far as $Pr[E_2]$ is regarded, we let $J_v$ be the number of non-rescuable pairs. It 
holds that
\begin{eqnarray}\label{eq:PE2+PropBound}
Pr[E_2]=Pr[J_v>0]\leq E[J_v],
\end{eqnarray}
where the last inequality follows from Markov's inequality. Using (\ref{eq:RescProbBound}), 
we get that
\begin{eqnarray}
E[J_v]&\leq& (1-\delta_k) d\nonumber \\
&\leq& \exp \left(-\frac{3(1+\epsilon)}{8\ln d}d^{\frac{\epsilon}{1+\epsilon}}\right)d
\leq \exp\left(-d^{\frac{\epsilon}{2(1+\epsilon)}}\right).\nonumber
\end{eqnarray}
Plugging the above inequality into (\ref{eq:PE2+PropBound}) we get that
\begin{eqnarray}
Pr[E_2]&\leq& \exp\left(-d^{\frac{\epsilon}{2(1+\epsilon)}}\right). \label{eq:PE2+PropBoundFound}
\end{eqnarray}

\noindent
Finally,  for $Pr[E_3]$ we work as follows:
\begin{equation}\label{eq:PE3+PropBound}
Pr[E_3]\leq Pr[E_3|\beta_v<100\ln d]+Pr[\beta_v\geq 100 \ln d].
\end{equation}
We let $M_v$ be the number of bad pairs  which are associated to  a partition with less than 
$d^{\frac45\frac{\epsilon-2}{1+\epsilon}}$ special pairs. Clearly, it holds that 
$$
Pr[E_3|\beta_v\leq 100\ln d]=Pr[M_v>0|\beta_v\leq 100\ln d].
$$
We remind the reader that $h^j_v$ denotes the number of $j$-special  pairs that appear 
in the partition that is associated to  the rescuable pair $(L^j_X,L^j_Y)$.
It holds that 
\begin{eqnarray}
Pr[h^j_v\leq d^{\frac45\frac{\epsilon-2}{1+\epsilon}}|\beta_v\leq 100\ln d]
&\leq & 
Pr\left [h^j_v\leq \frac{d^{\frac{\epsilon-2}{1+\epsilon}}}{1600\ln d}|\beta_v\leq 100\ln d\right]
\nonumber \\
%&\leq& \exp\left(-\frac{d^{\frac{\epsilon-2}{1+\epsilon}}}{6400\ln d}\right)  \qquad\qquad \mbox{[from (\ref{eq:hProbBound})]}
%\nonumber \\
&\leq & \exp\left(-d^{\frac{4}{5}\frac{\epsilon-2}{1+\epsilon}}\right). \qquad\qquad \mbox{[from (\ref{eq:hProbBound})]} \nonumber
\end{eqnarray}
It is direct that
\begin{eqnarray}
E[M_v|\beta_v\leq 100\ln d]&\leq& (100\ln d)
Pr[h^j_v\leq d^{\frac45 \frac{\epsilon-2}{1+\epsilon}}|\beta_v\leq 100\ln d]
\nonumber \\
&\leq & \exp\left(-d^{\frac{3}{5}\frac{\epsilon-2}{1+\epsilon}}\right).
\nonumber 
\end{eqnarray}
Using Markov's inequality we get that
\begin{displaymath}
Pr[M_v>0|\beta_v\leq 100\ln d]\leq E[M_v|\beta_v\leq 100\ln d]\leq \exp\left(-d^{\frac35\frac{\epsilon-2}{1+\epsilon}}\right).
\end{displaymath}
Plugging the above inequality and (\ref{eq:PE1+PropBoundFound})
to (\ref{eq:PE3+PropBound})  we get that
\begin{eqnarray}
Pr[E_3]&\leq &\exp\left(-d^{\frac35\frac{\epsilon-2}{1+\epsilon}}\right)+
d^{-250}
\leq 2d^{-250}.\label{eq:PE3+PropBoundFound}
\end{eqnarray}
Plugging (\ref{eq:PE1+PropBoundFound}), (\ref{eq:PE2+PropBoundFound}) and (\ref{eq:PE3+PropBoundFound}) into (\ref{eq:AEiRelation}) we get that 
$Pr[\mathbb{A}]\leq 5 d^{-250}.$
The proposition follows.

\subsection{Proof of Proposition \ref{prop:BoundsOnNumberLists}}
\label{sec:prop:BoundsOnNumberLists}

\noindent
The inequality in (\ref{eq:RescProbBound}) follows from the following
two lemmas.

\begin{lemma}\label{lemma:FreeColorProb}
Let $L^i_X$ be a list which belongs to a bad pair, for some $i\in [d]$.
For $k=(1+\epsilon)d/\ln d$ and for any colour $s\in [k]\backslash\{c,q\}$ 
it holds that
\begin{displaymath}
|Pr[s\notin L^i_X]-d^{-\frac{1}{1+\epsilon}}|\leq 
3d^{-\frac{2}{1+\epsilon}}.
\end{displaymath}
\end{lemma}

\begin{proof}
It holds that $c\in L^i_X$. Let $t$ be the number of the appearances of $c$ in the $L^i_X$.
Then, it holds that $Pr[s\notin L^i_X|t]=\left(1-\frac{1}{k-1}\right)^{d-t}$. The
random variable $t$ is binomially distributed with parameters
$1/(k-1)$ and $d$, conditional that it is positive. It is direct
that
\begin{eqnarray}
p_0=Pr[{\cal B}(1/(k-1),d)=0]&=&
\left(1-\frac{1}{k-1}\right)^d \leq \exp\left(-\frac{d}{k}\right)
\leq d^{-\frac{1}{1+\epsilon}}. \label{p_0:UpperB}
\end{eqnarray}

\noindent
Thus, it holds that
\begin{eqnarray}
 Pr[s\notin L^i_X]&=&\sum_{i=1}^d\left(1-\frac{1}{k-2}\right)^{d-i}Pr[t=i]
\nonumber \\
&=&\frac{1}{1-p_0}\sum_{i=1}^d{d \choose i}\left(\frac{1}{k-1}\right)^i
\left(1-\frac{1}{k-1}\right)^{d-i}\left(1-\frac{1}{k-2}\right)^{d-i}
\nonumber \\
&\leq & \frac{1}{1-d^{-\frac{1}{1+\epsilon}}}\left(1-\frac{1}{k-2}+\frac{1}{(k-1)(k-2)} \right)^d
\nonumber \\
&\leq&\left(1+2d^{-\frac{1}{1+\epsilon}}\right)
\exp\left(-\frac{d}{k}+\frac{d}{(k-2)^2} \right)  \hspace{0.6cm} \mbox{[as $1-x\leq e^{-x}$
and $d^{-\frac{1}{1+\epsilon}}<1/2$]}
\nonumber \\
&\leq& d^{-\frac{1}{1+\epsilon}}\left(1+2d^{-\frac{1}{1+\epsilon}}\right)
\left(1+2\frac{d}{(k-2)^2}\right) \hspace{0.7cm} \mbox{[as $e^{x}<1+2x$ for $0<x<0.1$]}
\nonumber \\
&\leq& d^{-\frac{1}{1+\epsilon}}\left(1+3d^{-\frac{1}{1+\epsilon}}\right).
\hspace{4.6cm} \mbox{[as $d/k^2=o_d(d^{-1/(1+\epsilon)})$]} \nonumber
\end{eqnarray}

\noindent
We get a lower bound on the $Pr[c\notin L^i_X]$ by working similarly. 
In particular, we have that
\begin{eqnarray}
Pr[s\notin L^i_X]&\geq & \frac{1}{1-p_0}\sum_{i=1}^d{d \choose i}\left(\frac{1}{k-1}\right)^i
\left(1-\frac{1}{k-1}\right)^{d-i}\left(1-\frac{1}{k-2}\right)^{d-i}
\nonumber \\
&\geq & \left(1-\frac{1}{k-2}+\frac{1}{(k-1)(k-2)} \right)^d -\left(1-\frac{1}{k-1}\right)^d\left(1-\frac{1}{k-2}\right)^d
\hspace{.7cm} \mbox{[as $\frac{1}{1-p_o}\geq 1$]}\nonumber \\
&\geq & \left(1-\frac{1}{k-2} \right)^d - \left(1-\frac{1}{k}\right)^{2d}
\nonumber \\
&\geq & \exp\left(-\frac{d}{k-2}\left(1-\frac{1}{k-2}\right)^{-1}\right)-\exp\left(-2d/k\right) 
\hspace{0.9cm} \mbox{[as $1-x\geq \exp(-\frac{x}{1-x})$ for $0<x<0.1$]} 
\nonumber \\
&\geq & \exp\left(-\frac{d}{k}-\frac{6d}{k^2} \right) -\exp\left(-2d/k\right) 
%\qquad \mbox{[as $\frac{1}{k}+\frac{i}{k}\leq \frac{1}{k-i}\leq \frac{1}{k}+\frac{2i}{k^2} $ for $0<i<k$]}
\nonumber \\
&\geq & d^{-\frac{1}{1+\epsilon}}\left(1-\frac{6d}{k^2}\right) 
-d^{-\frac{2}{1+\epsilon}} \nonumber \\
&\geq & d^{-\frac{1}{1+\epsilon}}\left(1-3d^{-\frac{1}{1+\epsilon}}\right).
\hspace{6cm} \mbox{[as $d/k^2=o_d(d^{-1/(1+\epsilon)})$]} \nonumber
\end{eqnarray}
The lemma follows.
\end{proof}

\begin{lemma}\label{lemma:AboutFreeColours}
Let $H_i$  denote the number of colours in $[k]\backslash\{c,q\}$ that do not appear in both 
lists of the bad pair $(L^i_X, L^i_Y)$. For $k=(1+\epsilon)d/\ln d$ and for any $y\in(0,1)$ 
it holds that

\begin{equation}%\label{eq:ChernoffFreeColours}
Pr\left[H_i\leq \frac{1-y}{3}\frac{(1+\epsilon)}{\ln d}d^{\frac{\epsilon-1}{1+\epsilon}}\right]
\leq
2\exp\left(-\frac{y^2}{6}\frac{(1+\epsilon)}{\ln d}d^{\frac{\epsilon-1}{1+\epsilon}}\right). \nonumber
\end{equation}
\end{lemma}

\begin{proof}
Since we have assumed that $(L^i_X, L^i_Y)$ is a bad pair, for $L^i_X$ we have that $c\in L^i_X$ and $q\notin L^i_X$,
while for $L^i_Y$ we have that $q\in L^i_Y$ and $c\notin L^i_Y$.

Let $f_X$, $f_Y$ be the number of colours that do not appear in
the lists $L^i_X$ and $L^i_Y$, respectively. 
Using Lemma \ref{lemma:FreeColorProb} we have that
\begin{eqnarray}
 E[f_X]&\geq &(k-2)d^{-\frac{1}{1+\epsilon}}\left(1-3d^{-\frac{1}{1+\epsilon}}\right)\nonumber \\
&\geq& 
(1+\epsilon)\frac{d^{\frac{\epsilon}{1+\epsilon}}}{\ln d}\left(1-4d^{-\frac{1}{1+\epsilon}}\right)\geq 
\frac{3}{4}\frac{(1+\epsilon)}{\ln d}d^{\frac{\epsilon}{1+\epsilon}} \label{eq:ExpctFreeCols}. 
\end{eqnarray}
Using a ball and bins argument, we can show that Chernoff bounds apply for $f_X$. 
In particular, for any $y\in (0,1)$ it holds that
\begin{eqnarray}
Pr[f_X\leq (1-y)E[f_X]] &\leq& \exp\left(-\frac{y^2}{2}E[f_X]\right)
\leq  \exp\left(-\frac{y^2}{2}\frac{3(1+\epsilon)}{4\ln d} d^{\frac{\epsilon}{1+\epsilon}}\right).
\qquad \mbox{[from (\ref{eq:ExpctFreeCols})]}.
\nonumber
\end{eqnarray}
Let the event $R_X=$``$f_X > \frac{3}{8}\frac{(1+\epsilon)}{\ln d}d^{\frac{\epsilon}{1+\epsilon}}$''.
\begin{eqnarray}
1-Pr[R_X]&=& Pr\left[f_X \leq \frac{3}{8}\frac{(1+\epsilon)}{\ln d}d^{\frac{\epsilon}{1+\epsilon}}\right] 
\leq  \exp\left(-\frac{3}{32}\frac{(1+\epsilon)}{\ln d}d^{\frac{\epsilon}{1+\epsilon}}\right), \label{eq:RXLowerBound}
\end{eqnarray}
where the last inequality follows from Chernoff bounds by setting $y=1/2$.

Any information for $f_X$ does not affect the 
distribution of the colourings in $L^i_Y$. 
This holds since the choice of colours in the two lists are {\em independent} with 
each other (Step 3 in Phase 1 of the coupling). 
That is, 
$E[H_i|f_X]\geq f_X\cdot d^{-\frac{1}{1+\epsilon}} 
\left(1-3d^{-\frac{1}{1+\epsilon}}\right)$.
Also, we get that
\begin{eqnarray}
 E[H_i|R_X]\geq \frac{3}{8}\frac{(1+\epsilon)}{\ln d}d^{\frac{\epsilon-1}{1+\epsilon}}(1-3d^{-\frac{1}{1+\epsilon}})
\geq \frac{1}{3}\frac{(1+\epsilon)}{\ln d}d^{\frac{\epsilon-1}{1+\epsilon}}.
\label{eq:Expctf_vCOnd}
\end{eqnarray}
Arguing in the same manner as above, we apply Chernoff bounds for $H_i$
and we get that for any $y\in (0,1)$
\begin{eqnarray}
 Pr[H_i\leq (1-y)E[H_i|R_X]|R_X]
&\leq&\exp\left(-\frac{y^2}{2}E[H_i|R_X]\right)\nonumber\\
&\leq&\exp\left(-\frac{y^2}{6}\frac{(1+\epsilon)}{\ln d}d^{\frac{\epsilon-1}{1+\epsilon}}\right).
\qquad \mbox{[from (\ref{eq:Expctf_vCOnd})]} \label{eq:ProbBf_vCOnd}
\end{eqnarray}
It holds that
\begin{eqnarray}
Pr\left[H_i \leq \frac{1-y}{3}\frac{(1+\epsilon)}{\ln d}d^{\frac{\epsilon-1}{1+\epsilon}}\right]&\leq&
Pr\left[H_i\leq \frac{1-y}{3}\frac{(1+\epsilon)}{\ln d}d^{\frac{\epsilon-1}{1+\epsilon}}|R_X\right]+(1-Pr[R_X]) \nonumber\\
&\leq& Pr[H_i\leq (1-y)E[H_i|R_X]|R_X] +1-Pr[R_X] \hspace{2.6cm} \mbox{[from (\ref{eq:Expctf_vCOnd})]}  \nonumber \\
&\leq& \exp\left(-\frac{y^2}{6}\frac{(1+\epsilon)}{\ln d}d^{\frac{\epsilon-1}{1+\epsilon}}\right)+
\exp\left(-\frac{3}{32}\frac{(1+\epsilon)}{\ln d}d^{\frac{\epsilon}{1+\epsilon}}\right).
\qquad \mbox{[from (\ref{eq:ProbBf_vCOnd}),(\ref{eq:RXLowerBound})]} \nonumber 
\end{eqnarray}
The lemma follows.
\end{proof}

\noindent 
Using Lemma \ref{lemma:AboutFreeColours}, where we set $y=1/2$, we get (\ref{eq:RescProbBound}), i.e. 
\begin{eqnarray}
\delta_k&\geq& 1-Pr\left[H_i\leq \frac{(1+\epsilon)}{6\ln d}d^{\frac{\epsilon}{1+\epsilon}}\right]
\geq 1-2\exp \left(
-\frac{(1+\epsilon)}{24\ln d}d^{\frac{\epsilon}{1+\epsilon}}
\right).
\nonumber 
\end{eqnarray}
Also, for proving (\ref{eq:BetaProbBound}) we use the following  lemma.

\begin{lemma}\label{lemma:NoOfBad}
For $k=(1+\epsilon)d/\ln d$, it holds that
\begin{equation}%\label{eq:ConcNoOfBad}
Pr\left[\beta_v \geq (1+x)\frac{d}{k-1} \right ]\leq d^{\left(-\frac{3\phi(x)}{4(1+\epsilon)} \right)}, \nonumber
\end{equation}
where $\phi(x)=(1+x)\ln(1+x)-x$, for $x>0$.
\end{lemma}
\begin{proof}
There are $d$ different pairs of lists and each of them is bad independently of the others. 
Let $p_{bad}$ be the probability for the pair  $(L^i_X,L^i_Y)$ to be bad. It suffices to have that 
$L_Y(i)=c$ while $q\in L^i_Y$. It holds that
\begin{eqnarray}
p_{bad}=\frac{1}{k-1}\left(1-\left(1-\frac{1}{k-1}\right)^d\right)
\leq \frac{1}{k-1}, \nonumber
\end{eqnarray}
as $\left(1-\left(1-\frac{1}{k-1}\right)^d\right)\leq 1$. By the linearity
of expectation we get that
\begin{eqnarray}\label{eq:ExpctBadUpper}
E[\beta_v]\leq dp_{bad}&\leq& d/(k-1).
\end{eqnarray}
Also, using Lemma \ref{lemma:FreeColorProb} we get that
\begin{displaymath}
p_{bad}\geq \frac{1}{k-1}\left(1-d^{-\frac{1}{1+\epsilon}}
\left(1+\frac{4d}{k^2}\right)\right)\geq 
\frac{3}{4k}.
\end{displaymath}
In turn, we get that
\begin{equation}\label{eq:ExpctBadLower}
E[\beta_v]\geq d p_{bad}\geq \frac{3\ln d}{4(1+\epsilon)}.
\end{equation}
Applying Chernoff bounds, for any $x>0$,  we have that 
\begin{displaymath}
Pr\left[\beta_v \geq (1+x) E[\beta_v]\right ]\leq \exp\left(-\phi(x)\cdot E[\beta_v]\right),
\end{displaymath}
where $\phi(x)=(1+x)\ln(1+x)-x$. %We get (\ref{eq:ConcNoOfBad}) 
The lemma follows by  plugging the bounds from (\ref{eq:ExpctBadUpper}) and (\ref{eq:ExpctBadLower})
into the above inequality. The lemma follows.
\end{proof}

\noindent
The next two lemmas show that  (\ref{eq:hProbBound}) holds.

\begin{lemma}\label{lemma:TypicalProbOfGood}
Let $(L^j_X,L^j_Y)$ be a rescuable pair and let $A_j$ be the set of indices of the pairs 
where we check for $j$-special lists.  Assume that $A_j$ is non empty.  
Let $k= (1+\epsilon)d/\ln d$. For any $i\in A_j$, it holds that
\begin{displaymath}
Pr[\textrm{$L^i_Y$ is special w.r.t. $L^j_X$}| H_j]\geq %\frac{10}{9}d^{-\frac{3}{1+\epsilon}},
\frac{H_j}{k}d^{-\frac{1}{1+\epsilon}},
\end{displaymath}
where $H_j$ is the number of colours that do not appear in both $L^j_X,L^j_Y$.
\end{lemma}
\begin{proof}
Since $(L^j_X,L^j_Y)$  is rescuable, it means that $L_X(j)=c$ and $q\in L^j_X$.
Also, there is non-empty set of colours $U_i\in [k]\backslash\{c,q\}$ which contains 
colours that do not appear in $L^j_X\cup L^j_Y$. So as to have $L^i_Y$ special special
w.r.t. $L^j_X$, it should hold that  $L_Y(i)\in U_j$ and either of 
the following two holds 
A) $q\in L^i_Y$ and $c\notin L^i_Y$ or 
B) $q\notin L^i_Y$ and $c\in L^i_Y$.
Let the event $\mathbb{Q}=$ ``$L_Y(i)\in U_j$''.  It holds that
\begin{equation}\label{eq:ProbGood}
\varrho_k \geq  2\frac{H_j}{k-2}Pr[q\notin L^j_Y|  c\in L^j_Y,\mathbb{Q}]
Pr[c\in L^j_Y|\mathbb{Q}],
\end{equation}
since $H_j=|U_j|$. Also, working as in the proof of Lemma \ref{lemma:FreeColorProb} 
we get that
\begin{eqnarray}
Pr[c\in L^i_Y|\mathbb{Q}]&\geq& 1-d^{-\frac{1}{1+\epsilon}}\left(1+3d^{-\frac{1}{1+\epsilon}}\right),
\label{eq:GoodSecondTermLB} \\
Pr[q\notin L^j_Y|  c\in L^j_Y,\mathbb{Q}]&\geq& \frac{3}{4}d^{-\frac{1}{1+\epsilon}}.
\label{eq:qnotinLiY}
\end{eqnarray}
The lemma follows by substituting the bounds from (\ref{eq:GoodSecondTermLB})
and  (\ref{eq:qnotinLiY}) %and (\ref{eq:GoodFirstTermLB})  
into  (\ref{eq:ProbGood}).
%we get
%\begin{eqnarray}
%\varrho_k&\geq &\frac{3}{2}d^{-\frac{3}{1+\epsilon}}\left(1-7d^{-\frac{1}{1+\epsilon}}\right)
%\left( 1-d^{-\frac{1}{1+\epsilon}}\left(1+3d^{-\frac{1}{1+\epsilon}}\right)\right)
%\geq \frac{10}{9}d^{-\frac{3}{1+\epsilon}}. 
%\nonumber
%\end{eqnarray}
\end{proof}

\begin{lemma}
%Consider a random $k$-colouring of $T$ for $k=(1+\epsilon)d/\ln d$. 
Let $h^j_v$ be the number of the non-bad  pairs that correspond to the rescuable 
pair $(L^j_X,L^j_Y)$. For $k=(1+\epsilon)d/\ln d$ and fixed $c>0$, it holds that
\begin{displaymath}
Pr\left[h^j_v\leq \frac{d^{\frac{\epsilon-2}{1+\epsilon}}}{16c\ln d}
\left| \right. \beta_v\leq c\ln d\right]
\leq 2\exp\left(-\frac{d^{\frac{\epsilon-2}{1+\epsilon}}}{64c\ln d}
\right).
\end{displaymath}
\end{lemma}
\begin{proof}
The number of lists that are associated to each rescuable pair depends on the actual number of bad lists. 
Conditioning that the number of bad pairs $\beta_v\leq c\ln d$,  for some fixed $c>0$, the rescuable 
pair $(L^j_X,L^j_Y)$, is assigned a set of at least  $\lfloor\frac{d}{c \ln d}-1\rfloor$ non-bad lists.
Let $H_j$ denote the number of colours that do not in $L^j_X\cup L^j_Y$.
Let the event ${\cal H}=$``$H_j>\frac{1+\epsilon}{6\ln d}d^{\frac{\epsilon-1}{1+\epsilon}}$''.
From Lemma \ref{lemma:AboutFreeColours} we get that
\begin{eqnarray}\nonumber
1-Pr[{\cal H}]= Pr\left[H_j\leq \frac{1+\epsilon}{6\ln d}d^{\frac{\epsilon-1}{1+\epsilon}}\right]
 \leq 2\exp\left (-\frac{1+\epsilon}{24\ln d}d^{\frac{\epsilon-1}{1+\epsilon}}\right ).
\end{eqnarray}

%Clearly, it holds that $Pr[{\cal H}]\geq 1-2\exp\left (-\frac{1+\epsilon}{24\ln d}d^{\frac{\epsilon-1}{1+\epsilon}}\right )$.

\noindent
From Lemma \ref{lemma:TypicalProbOfGood} we get that
\begin{eqnarray}
E\left[h^j_v|{\cal H}, \beta_v\leq c\ln d\right]&\geq& 
%d^{-\frac{3}{1+\epsilon}}\left(\frac{d}{c\ln d}- 2\right) \geq 
\frac{d^{\frac{\epsilon-2}{1+\epsilon}}}{8c\ln d}.\nonumber
\end{eqnarray}
%Also, note that given the rescuable lists, each of the remaining lists
%are special independently of the other lists. Thus, 
We can apply Chernoff bounds and get the following 
\begin{eqnarray}
Pr\left[h^j_v\leq (1-y) \frac{d^{\frac{\epsilon-2}{1+\epsilon}}}{8c\ln d}
\left| \right. {\cal H}, \beta_v\leq c\ln d\right]
&\leq& \exp\left(-\frac{y^2}{16c}\frac{d^{\frac{\epsilon-2}{1+\epsilon}}}{\ln d}
\right).
\nonumber
\end{eqnarray}
%%The lemma follows by Setting above $y=1/2$ 
From the law of total probability it holds that
\begin{eqnarray}
 Pr\left[h^j_v\leq (1-y) \frac{d^{\frac{\epsilon-2}{1+\epsilon}}}{8c\ln d} \left| \right. \beta_v\leq c\ln d\right]
 &\leq& Pr\left[h^j_v\leq (1-y) \frac{d^{\frac{\epsilon-2}{1+\epsilon}}}{8c\ln d}\left| \right. {\cal H}, \beta_v\leq c\ln d\right]
+Pr[{\cal H}^c|\beta_v \leq  c\log d] \nonumber \\
&\leq& \exp\left(-\frac{y^2}{16c}\frac{d^{\frac{\epsilon-2}{1+\epsilon}}}{\ln d}\right)+
2\exp\left (-\frac{1+\epsilon}{24\ln d}d^{\frac{\epsilon-1}{1+\epsilon}}\right). \nonumber
\end{eqnarray}
We used the fact that the events ${\cal H}$ and ``$\beta_v < c\ln d$'' are independent
with each other. The lemma follows by setting $y=1/2$.
\end{proof}

\section{Proof of Proposition \ref{prop:PlantingBits}}\label{sec:planting}

A way of constructing ${\cal S}$, which is equivalent to the one described in 
Remark \ref{remark:CouplingString}, is the following one: Consider some sufficiently 
large positive integer $l\ll N$. We construct ${\cal S}$ in rounds. Assume that after 
round $i-1$ we have constructed  ${\cal S}$ up to some column $t$, for some $t \ll N$. 
Additionally, let $X_t=\sum_{j=1}^{t}{\cal S}(1,j)-{\cal S}(2,j)$. Then, during the
round $i$ we proceed as described in the  following paragraph.

If $X_t=0$, then we use identical coupling for ${\cal S}(1,j), {\cal S}(2,j)$ for all $t<j\leq N$.
If $X_t>0$, then we consider $X_t$ many sets of columns of ${\cal S}$ whose values has not been 
set yet. Each of these $X_t$ many sets contains at most $l$ columns. More specifically, the first
set $R^i_1$ starts from column $t+1$ up  to column $T$, the value of $T$ will be defined in what 
follows. We set the values in each column $j\in R^i_1$  by coupling ${\cal S}(1,j)$ ${\cal S}(2,j)$
such that ${\cal S}(1,j)=1-{\cal S}(2,j)$.
%\footnote{We remind the reader that ${\cal S}(1,j)=1$  with probability $1/2$ and ${\cal S}(1,j)=0$ with the remaining probability.} 
%
$T$ is either the first time that $\sum_{j=t+1}^T {\cal S}(1,j)-{\cal S}(2,j)=-1$
or if this is not possible up to column $t+l$, then we have $T=t+l$. Continue with 
the second set of columns $R^i_2$ \footnote{$R^i_2$ starts from the column $T+1$} 
and so on. Round $i$ ends after having finished with all these $X_t$ sets of columns. Then
we continue in the same manner with the round $i+1$.

For each set of columns $R^i_j$, ($R^i_j$ is submatrix of ${\cal S}$), we have the following 
lemma which is going to be useful in the proof of Proposition \ref{prop:PlantingBits}.

\begin{lemma}\label{lemma:RandomWalkDecay}
Let $l\geq 10$, the maximum number of columns of $R^i_j$. If the entries 
are such that $R^i_j(1,s)\neq R^i_j(2,s)$ for any column $s$ of $R^i_j$, 
then it holds that 
\begin{displaymath}
E\left[1+\sum_{t=1}^{T} R^i_j(1,t)-R^i_j(2,t)\right]\leq \frac{2.3}{\pi},
\end{displaymath}
where $T$ is the actual number of columns of $R^i_j$.
\end{lemma}
\begin{proof}
For every $t$ it holds that $R^i_j(1,t)-R^i_j(2,t)$ is equal to
$-1$ with probability $1/2$ or it is equal to $1$ with probability
$1/2$. It is direct to see that the partial sums  
$W_s=\sum_{t=1}^{s} R^i_j(1,t)-R^i_j(2,s)$, for $s\leq T$ constitute
a symmetric random walk on the integers which starts from position 
zero and stops either when it hits $-1$ or after $l$ steps, whatever 
happens first. We can simplify the analysis and remove the dependency 
from the random variable $T$, by assuming that $W_s$ continues always for
$l$ steps and the state $-1$ is absorbing. Then, the lemma follows
by just computing $E[W_{l}+1]$. In particular, we have that
\begin{equation}\label{eq:ExpectedPosition}
E[W_{l}+1]= E[W_{l}+1|W_{l}\neq -1]\cdot Pr[W_{l}\neq -1].
\end{equation}

\noindent
Let $\cal T$ be the step that $W_t$ hits $-1$ for first time. By the
{\em Reflection Principle} we have that for any nonnegative integer $i$ it holds that
\begin{equation}\label{eq:ProT=2i+1}
Pr[{\cal T}=2i+1]=2^{-(2i+1)}\frac{{2i \choose i}}{i+1}.
\end{equation}
It is direct that the $W_t$ cannot be $-1$ for $t$ even, i.e. $Pr[{\cal T}=2i]=0$, for 
every positive integer $i$. It is direct to see that it holds that
\begin{displaymath}
Pr[W_{l}=-1]=Pr[{\cal T}\leq l]=1-
\sum_{i> \lfloor (l-1)/2\rfloor} 2^{-(2i+1)}\frac{{2i \choose i}}{i+1}.
\end{displaymath}
To this end we use Stirling approximation, i.e. for a sufficiently large $n$ it holds 
that $n!=\sqrt{2\pi n}\left(\frac{n}{e}\right)^n e^{\lambda_n}$, with $\frac{1}{12n+1}
\leq \lambda_n\leq \frac{1}{12n}$. Then we have that
\begin{eqnarray}
\sum_{i> \lfloor (l-1)/2\rfloor} 2^{-(2i+1)}\frac{{2i \choose i}}{i+1} 
&\leq& \frac{1}{2\sqrt{\pi}} \sum_{i> \lfloor (l-1)/2\rfloor} \frac{1}{i^{3/2}} 
\leq \sqrt{\frac{2}{\pi l}}\nonumber.
\end{eqnarray}
Thus, we get that
\begin{equation}\label{eq:ProbHit-1}
Pr[W_{l}=-1]
\geq 1-\sqrt{\frac{2}{\pi l}}.
\end{equation}
On the other hand, it is direct to see that given that the walk $W_t$ does not hit $-1$ it 
is just a random walk on the positive integers and it is a folklore result that 
\begin{equation}\label{eq:ExpectationForPositive}
E[Z_{l}|Z_{l}\neq -1]\leq \sqrt{\left(\frac{2}{\pi}l\right)}\cdot \left(1+\frac{3}{2l}\right).
\end{equation}
The lemma follows by plugging (\ref{eq:ProbHit-1}) and (\ref{eq:ExpectationForPositive}) into
(\ref{eq:ExpectedPosition}) and taking $l\geq 10$.
\end{proof}

\noindent
\begin{propositionproof}{\ref{prop:PlantingBits}}
Consider the revelation of the values of the matrix ${\cal S}$ we gave above. Let $t_i$ be the 
index of the column we have revealed  up to round $i$. I.e. at round $i+1$ we check whether 
$X_{t_i}=\sum_{j=1}^{t_i}{\cal S}(1,j)-{\cal S}(2,j)$ is zero or not. Let $l$ the maximum number 
of columns in each submatrix $R^i_j$ be equal to 10.

Given $X_{t_i}$ and assuming that the coupling continuous, i.e.  $t_i$ the number of columns we 
have revealed so far is  much smaller  than $N$, we show that it holds that 
\begin{equation}\label{eq:DecayRandomWalk}
E[X_{t_{i+1}}|X_{t_i}]\leq \frac{2.3}{\pi}X_{t_i}.
\end{equation}
However, before showing the above let us see which are its 
consequences. Taking the average from both sides, we get
\begin{eqnarray}
E[X_{t_i}]&\leq& \frac{2.3}{\pi} E[X_{t_{i-1}}] \leq \left(\frac{2.3}{\pi}\right)^i,
\nonumber 
\end{eqnarray}
since $X_{t_1}=1$ (it always holds that ${\cal S}(1,1)-{\cal S}(2,1)=1$).
It is also direct to see that it always holds that $X_{t_i}\leq l
\cdot X_{t_{i-1}}\leq l^i$. That is, in round $i$ we will need to reveal at most
$l^i$ columns of the matrix. This fact implies that the maximum $j$  which 
satisfies the condition that $\sum_{t=0}^jl^t\leq N$ is a lower bound for
the number of rounds we can have. Direct calculations  suggest that 
the number of rounds $j_0\geq \frac{99}{100}\frac{\ln N}{\ln l}=0.43\ln N$, since $l=10$.
Clearly, the proposition follows once we show (\ref{eq:DecayRandomWalk}).
For this we are going to use Lemma \ref{lemma:RandomWalkDecay}. 
Notice that given that at round $i$ we have 
$X_{t_i}=\left |\sum_{j=1}^{t_i}\left({\cal S}(1,j)-{\cal S}(2,j)\right)\right|$,
for $X_{i+1}$ the following holds:
$$
X_{t_{i+1}}=\sum_{s=0}^{X_{t_i}}\left(1+\sum_{j=1}^{T_s}R^{i}_s(1,j)-R^{i}_s(2,j)\right),
$$
where $T_s$ is the length of the submatrix $R^i_s$. From Lemma 
\ref{lemma:RandomWalkDecay} we have that for any $i,s$ it holds
$$
E\left[1+\sum_{j=1}^{T_s}R^{i}_s(1,j)-R^{i}_s(2,j) \right] \leq \frac{2.3}{\pi}.
$$ 
Combining the above two relations  and by linearity of expectation we get that
\begin{eqnarray}
E[X_{t_{i+1}}|X_{t_{i}}]&=&\sum_{s=1}^{X_{t_{i}}}
E\left[1+\sum_{j=1}^{T_s}R^{i}_s(1,j)-R^{i}_s(2,j)\right]
\leq \frac{2.3}{\pi}X_{t_{i}}. \nonumber 
\end{eqnarray}
The proposition follows.
\end{propositionproof}

\section{Rest of the proofs}

\subsection{Proof of Lemma \ref{lemma:RescuabelVsGood}}\label{sec:lemma:RescuabelVsGood}
Since $(L_1,L_2)$ is a rescuable (thus bad) pair, we have the following information 
for the lists. For $L_1$ we know that the colour $q\in L_1$, and $c\notin L_1$.
For $L_2$, we know that $q\in L_2$, $c\notin L_2$. 
%Since $(L_1,L_2)$ is rescuable, 
Also, there is a non-empty set of colours $U\subseteq [k]\backslash\{c,q\}$
such that for each $c'\in U$ it holds that $c'\notin L_1\cup L_2$.
Finally, since $L_3$ good with respect to $L_1$,  this implies that $s\in U$
while $q\in L_3$ and $c\notin L_3$. 

Let the event $A=$``$L_3 \textrm{ is good w.r.t. $L_1$}$''.
For any $S\in [k]^d$ it holds that
\begin{displaymath}
Pr[L_3=S|A]=\lambda_{s}(S|B),
\end{displaymath}
where $B=$``there exists $t\in [d]$ such that $S(t)=q$ and 
there is no $t\in [d]$ such that $S(t)=c$''.

Let $Q=|U|$. It suffices to show that, 
\begin{equation}\label{eq:TargetGoodVsRescuableA}
Pr[L_1=S|c\notin L_1, q\in L_{1}, s\in U, Q>0 ]=\lambda_{s}(S|B).
\end{equation}
Clearly we have that
\begin{eqnarray}
Pr[L_1=S|c\notin L_1, q\in L_{1}, s\in U, Q>0]&=&
\frac{Pr[L_1=S, c\notin L_1, q\in L_{1}, s\in U, Q>0]}{Pr[c\notin L_1, q\in L_{1}, s\in U, Q>0]}\nonumber\\
&=&\frac{Pr[L_1=S, c\notin L_1, q\in L_{1}, s\in U]}
{Pr[c\notin L_1, q\in L_{1}, s\in U ]}\nonumber\\
&=& 
Pr[L_1=S| s,c\notin L_1, q\in L_{1}].
\label{sec:lemma:RescuabelVsGoodB}
\end{eqnarray}
In the  penultimate  derivation we eliminated the event $Q>0$ from both probability terms,
in the nominator and denominator, since whenever $s\in U$ holds it also holds that $Q>0$. 
Then, it is straightforward that the r.h.s. of (\ref{sec:lemma:RescuabelVsGoodB})
is equal to $\lambda_{s}(S|B)$, i.e. (\ref{eq:TargetGoodVsRescuableA}) holds.

\subsection{Proof of Lemma \ref{lemma:GoodVsFail}}\label{sec:lemma:GoodVsFail}

The lemma follows by just examining the information we have for each of the four lists. 
As far as the $i$-good pair $(L^t_X,L^t_Y)$ is concerned we have the following: 
$L_X(t)$ is distributed uniformly at random among the colours $[k]\backslash\{c,q\}$ that do not appear
in $L^{i}_X \cup L^{i}_Y$, while $c\notin L^t_X$ and $q\in L^t_X$.
Also, $L_Y(t)$ is distributed uniformly at random among the colours $[k]\backslash\{c,q\}$ that 
do not appear in $L^{i}_Y \cup L^{i}_X$, while $q\notin L^t_Y$ and $c\in L^t_Y$.

As far as the $i$-fail pair $(L^s_X,L^s_Y)$ is concerned we have the following: $L_X(s)$ is 
distributed uniformly at random among the colours $[k]\backslash\{c,q\}$ that do not appear 
in $L^{i}_X \cup L^{i}_Y$ while $q\notin L^s_X$ and $c\in L^s_X$. Additionally, $L_Y(s)$ is distributed 
uniformly at random among the colours $[k]\backslash \{c,q\}$ that do not appear in $L^{i}_X\cup L^{i}_Y$ 
while $c\notin L^s_Y$ and $q\in L^t_Y$.

Thus, we can couple identically $L_X(t)$ with $L_Y(s)$ and $L_X(s)$ with $L_Y(t)$. Then, it is 
direct that we can couple identically $L^t_X$ with $L^s_Y$ and $L^s_X$ with $L^t_Y$.

\end{document}